\documentclass[aps,pra,twocolumn,floatfix,nofootinbib,superscriptaddress,graphics,longbibliography]{revtex4-2}

\usepackage{bm,graphicx,mathrsfs,amsmath,amssymb,mathtools,makecell,bbm, amsthm,dsfont,color,times,txfonts,nicefrac,framed,float,enumitem,tikz,physics,wrapfig,amsfonts,tcolorbox}
\usepackage{physics}
\usepackage{nicematrix}
\usepackage[colorlinks=true,linkcolor=equationcolor,citecolor=teal]{hyperref}

\hypersetup{urlcolor=teal}
\definecolor{equationcolor}{RGB}{222,94,100}
\definecolor{boxcolor}{RGB}{215,215,253}
\definecolor{block2c}{RGB}{215,185,212}
\definecolor{block3c}{RGB}{216,155,172}
\definecolor{block4c}{RGB}{216,125,132}

\definecolor{frenchbeige}{rgb}{0.65, 0.48, 0.36}

\makeatletter
\def\blfootnote{\gdef\@thefnmark{}\@footnotetext}
\makeatother

\newcommand{\AB}{\ms{AB}}
\newcommand{\ketbraa}[2]{\ensuremath{\left|#1\right\rangle\!\!\left\langle#2\right|}}

\renewcommand{\v}[1]{\ensuremath{\boldsymbol #1}}
\newcommand{\ms}[1]{\textsf{#1}}

\newcommand{\1}{\mathbbm{1}}
\newcommand{\iden}{\mathbbm{1}}

\newcommand{\E}[1]{\mathcal{E}}

\def\E{ {\cal E} }

\def\S{ {\cal S} }

\def\T{ {\cal T} }

\def\Bd{{\operatorname{Bd}}}

\usepackage[normalem]{ulem}

\newtheorem{thm}{Theorem}

\newtheorem{lem}[thm]{Lemma}
\newtheorem{prop}[thm]{Proposition}

\newtheorem{defn}{Definition}
\graphicspath{{Figures/}}

\begin{document}

\title{Entanglement generation from athermality}

\author{A. de Oliveira Junior$^*$}
	\affiliation{Center for Macroscopic Quantum States bigQ, Department of Physics,
Technical University of Denmark, Fysikvej 307, 2800 Kgs. Lyngby, Denmark}

\author{Jeongrak Son$^*$}
	\affiliation{School of Physical and Mathematical Sciences, Nanyang Technological University,
637371, Singapore}

\author{Jakub Czartowski}
	\affiliation{Doctoral School of Exact and Natural Sciences, Jagiellonian University, 30-348 Kraków, Poland}
	\affiliation{Faculty of Physics, Astronomy and Applied Computer Science, Jagiellonian University, 30-348 Kraków, Poland}
 
	\author{Nelly H. Y. Ng}
	\affiliation{School of Physical and Mathematical Sciences, Nanyang Technological University,
637371, Singapore}
 
\date{\today}

% ------------------------------------------------
%   ABSTRACT
% ------------------------------------------------

\begin{abstract}

We investigate the thermodynamic constraints on the pivotal task of entanglement generation using out-of-equilibrium states through a model-independent framework with minimal assumptions. 
We establish a necessary and sufficient condition for a thermal process to generate bipartite qubit entanglement, starting from an initially separable state. 
Consequently, we identify the set of system states that \emph{cannot} be entangled, when no external work is invested. 
In the regime of infinite temperature, we analytically construct this set; while for finite temperature, we provide a simple criterion to verify whether any given initial state is or is not entanglable.
Furthermore, we provide an explicit construction of the \emph{future thermal cone of entanglement} — the set of entangled states that a given separable state can thermodynamically evolve to. 
We offer a detailed discussion on the properties of this cone, focusing on the interplay between entanglement and its volumetric properties. 
We conclude with several key remarks on the generation of entanglement beyond two-qubit systems, and discuss its dynamics in the presence of dissipation.

\end{abstract}

\maketitle

%%%%%%%%%%%% INTRODUCTION %%%%%%%%%%%%

\section{Introduction}\label{Sec:introduction}
\blfootnote{$ ^*$\hspace{1pt}These authors contributed equally to this work.}

\setlength\intextsep{0pt}
\setlength{\columnsep}{1pt}%
\begin{wrapfigure}{l}{0.07\textwidth}
    \includegraphics[width=0.07\textwidth]{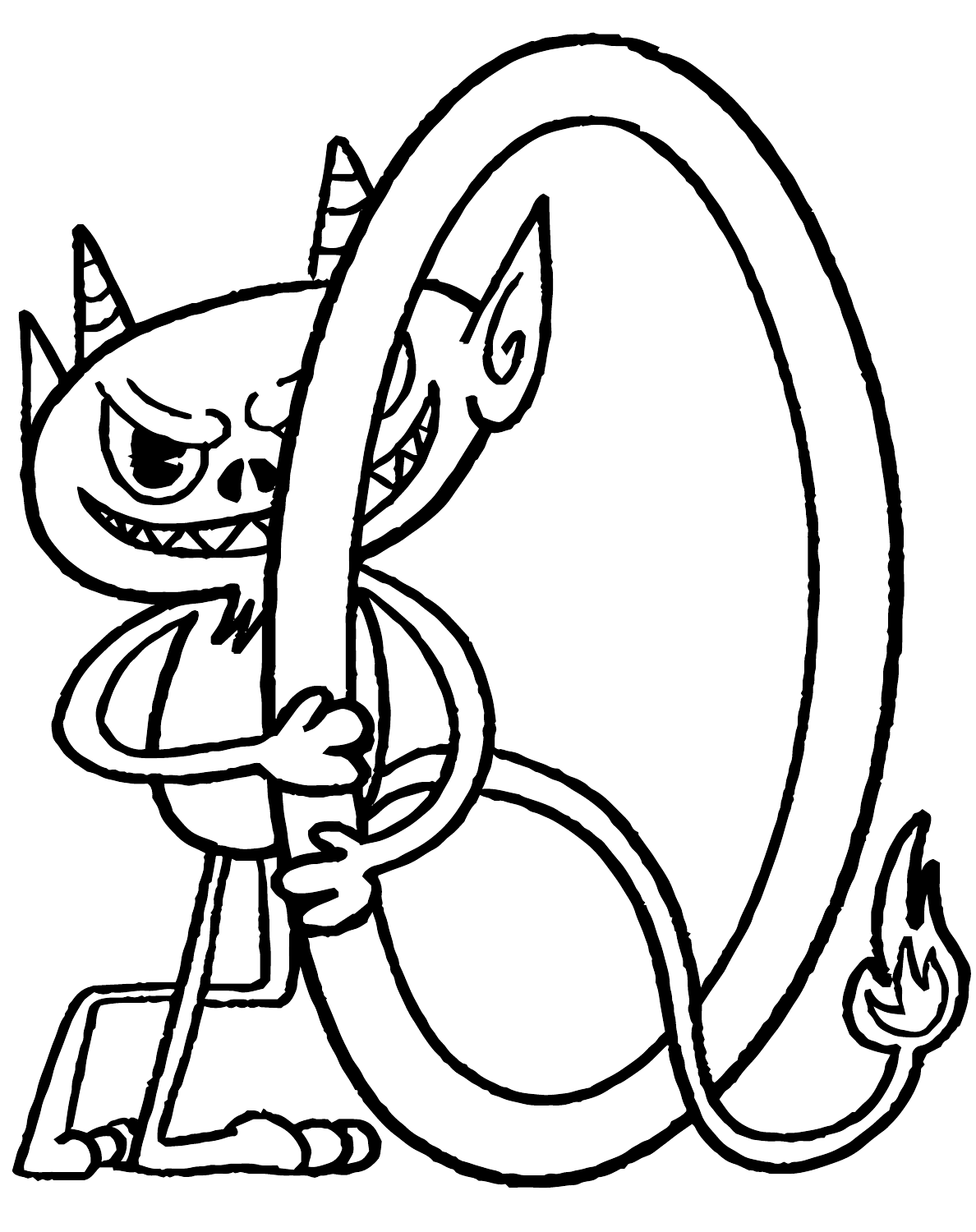}
\end{wrapfigure}
\setlength\intextsep{0pt}
\noindent 
uantum dynamics and thermodynamics have been burgeoning areas of study, particularly in the domain of quantum information science~\cite{horodecki2009quantum,kosloff2013quantum,Goold2016,vinjanampathy2016quantum,deffner2019quantum}. 
Upon moving from the macroscopic world to the quantum realm, it is natural to ask how thermodynamics affects quantum information processing~\cite{oppenheim2002thermodynamical,hovhannisyan2013entanglement,huber2015thermodynamic,perarnau2015extractable}. A thermodynamics-inspired paradigm for quantifying quantum correlations has led to the conclusion that one cannot extract the maximum work from an entangled state when it is shared by two separated parties~\cite{oppenheim2002thermodynamical, Alicki_CompPassive}. Conversely, correlations shared between quantum systems allow for optimal work storage~\cite{perarnau2015extractable}. 
In tandem with these inquiries, the fundamental energetic limitations on creating both classical and quantum correlations have been thoroughly investigated~\cite{Verstraete2001,huber2015thermodynamic,PhysRevE.93.042135, Verstraete2001,PhysRevA.62.022310}, assuming closed system dynamics.

The task of correlation engineering becomes increasingly interesting, in the more realistic scenario where quantum systems interact with a thermal environment, exchanging energy and entropy. Counterintuitively, it has been shown that steady-state entanglement can be generated~\cite{Plenio2002,hartmann2006steady,hartmann2007entanglement,bellomo2013steady} even when no external agents are involved~\cite{brask2015autonomous}. In particular, one can generate a maximally entangled steady-state between two qubits deterministically using a third ancillary qubit~\cite{khandelwal2024maximal}. These results suggest a paradigm shift where certain dissipative processes, rather than being detrimental to quantum correlations, might harness environment interactions for thermodynamically stable entanglement generation. As a result, a natural question arises -- \emph{what are the thermodynamic limitations for quantum correlation and entanglement generation?}

A powerful toolkit for answering this question is the \emph{resource-theoretic} approach~\cite{gour2015resource,PhysRevLett.111.250404}. One may initially think that the entanglement resource theory of local operations and classical communication (LOCC) should be the relevant framework for the entanglement manipulation problem. However, LOCC operations are naturally motivated only after entanglement has been distributed; if one is instead at the preparation phase in a local lab, prior to entanglement distribution, then it is more reasonable to consider a resource theory that contains entangling operations.
If entangling operations are allowed, the main limiting factor must come from quantum control and energetics~\cite{Janzing2000,horodecki2013quantumness, brandao2015second,Lostaglio2019}, which leads us to the resource theory of quantum thermodynamics. This model-independent framework allows one to investigate general thermodynamic processes by relying solely on the assumptions of total energy conservation and thermality of environments. Its broad scope has allowed us to establish bounds on the size of the bath required for the thermalization of a many-body system~\cite{Sparaciari2021}, bridge the gap between Markovian and non-Markovian thermal processes~\cite{PRXQuantum.4.040304,jeongraknelly}, and uncover optimal heat-bath algorithmic cooling strategies~\cite{alhambra2019heat}, among other significant results~(see~\cite{lostaglio2019introductory, torun2023compendious} for recent reviews).

In this work, we systematically probe the capability of using thermal processes to generate entanglement when given access to some initial state. In the particular case of two qubits,
we derive necessary and sufficient conditions for the existence of a thermal process that transforms a given separable initial state into a final state that is entangled. These conditions depend only on the energetic populations of the two-qubit system. A key insight that led to this finding, is the observation that a generic entangling strategy can be decomposed into two independent procedures: one that is explicitly achieved by open system dynamics (interaction with heat bath that results in changes of energetic population), and another closed unitary operation, i.e. an optimal rotation within the global energy subspace. The consistency of our results is confirmed by the problem addressed in Refs.~\cite{Verstraete2001,PhysRevA.62.022310}, where the authors investigate two-qubit mixed states and determine unconstrained global unitary operations that generate maximal entanglement. We approach the considered problem with additional restrictions on the alllowed operations, motivated by cornerstone findings on thermodynamic resource theories. 
To be specific, we assume that entangling operations must fall into the set of thermal operations~\cite{Janzing2000, horodecki2013fundamental}. 

Using the above framework we provide an 
explicit construction of the set of separable states that can become entangled. The degree of entanglement in this set can be assessed either by the volume of the set or by their negativity. In this regard, our results show that a Bell pair can always be thermodynamically produced at any temperature regime, provided that the system is initialized in its pure excited state. We also discuss the generation of entanglement in the presence of dissipation, examining the behavior of the volume of the future thermal cone of entanglement under Markovian thermal processes, and comparing it with the closed dynamics case. Finally, we offer some remarks on how our results can be extended to systems beyond two qubits.

The paper is organized as follows: We start by motivating how thermal resources can be used to generate entanglement in Section~\ref{Sec:motivating-example}, using an experimentally relevant example involving two spins interacting with a cavity mode. After this warm-up, we set the scene in Section~\ref{Sec:setting-the-scene} by recalling the resource-theoretic approach to quantum thermodynamics. In Section~\ref{Sec:entanglability}, we rigorously define the main quantifiers used for entanglement generation under thermal operations in arbitrary dimensions. 
Section~\ref{Sec:ent-criteria-2-qubit} focuses on 
the two-qubit system, and we derive necessary and sufficient conditions for entanglement preparation,
and proceeds to quantify the entanglement generated and discuss its applicability. We discuss in Section \ref{sec:discussions} how the task of entanglement generation changes for continuous-time thermal processes and high-dimensional systems, and conclude with an outlook in Section.~\ref{Sec:outlook}. 

%%%%%%%%%%%% MOTIVATING EXAMPLES %%%%%%%%%%%%

\section{Motivating examples}
\label{Sec:motivating-example}

Let us begin with a warm-up example that gives a flavor of the main investigation. Consider two qubits $\mathsf{A}$ and $\mathsf{B}$, with the former prepared in the ground state and the latter in the excited state. The joint state of the composite system is given by $\rho_{\mathsf{AB}} = \ketbra{01}{01}$. Furthermore, assume that each qubit is described by the Hamiltonian $H_{0} = E\ketbraa{1}{1}$, the composite $\mathsf{AB}$ is described by a Hamiltonian \mbox{$H = H_{0}\otimes\mathbbm{1}+\mathbbm{1}\otimes H_{0}$}. This Hamiltonian gives rise to the degenerate energy subspace $\mathcal{V}_{E} = \mathrm{span}\{\ket{01},\ket{10}\}$, associated with energy $E$.
Then any unitary acting non-trivially only on $\mathcal{V}_{E}$ is energy-preserving, including
\begin{equation}\label{eq:01_to_Bell}
		U = \frac{1}{\sqrt{2}}\begin{pmatrix}
			1 & 1 \\ 1 & -1
		\end{pmatrix} \oplus\1_{\mathcal{H}\setminus\mathcal{V}_{E}},
\end{equation} 
that maps product states $\ket{01}$ and $\ket{10}$ into Bell states $\ket{\phi^{+}}$ and $\ket{\phi^{-}}$, where $\ket{\phi^{\pm}}=\frac{1}{\sqrt{2}}(\ket{01}\pm\ket{10})$.
In particular, when the initial state is either $\ket{01}$ or $\ket{10}$, this transformation can be implemented with Jaynes-Cummings interactions only~\cite{Jaynes1963}
\begin{align}\label{eq:JC_Hamiltonian}
		H_{\mathrm{JC}} = \sigma^{+}_{\ms{A(B)}}\otimes a_{\ms{R}} + \sigma^{-}_{\ms{A(B)}}\otimes a^{\dagger}_{\ms{R}},
\end{align}
albeit indirectly by introducing a mediating bosonic mode $\ms{R}$~\cite{Cirac1994_catalysis}. We use the notation $\sigma^{+} = \ketbraa{1}{0}$ and $\sigma^{-} = \ketbraa{0}{1}$ for each qubit $\ms{A}$ or $\ms{B}$, while $a_{\ms R}$ and $a_{\ms R}^{\dagger}$ are the annihilation and creation operators for the bosonic mode.
\begin{figure}[t]
	\centering
	\includegraphics{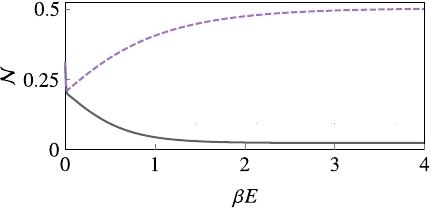}
	\caption{\textbf{Negativity as a function of $\beta  E$.}
		Preprocessing the spin system by coupling to a thermal bath opens up the possibility of generating entanglement through energy-preserving unitaries that encode the laws of thermodynamics. We plotted the maximal negativity~\cite{Vidal2002} [defined in Eq.~\eqref{Eq:negativity}] obtainable by a two-step process \emph{i)} transforming one of the spins via Jaynes-Cummings interaction with a thermal bosonic mode and \emph{ii}) applying the entangling unitary, Eq.~\eqref{eq:01_to_Bell}. The $x$-axis is the product of inverse environment temperature $\beta$ and energy gap $ E$. The black solid curve represents the resulting negativity from a state initially prepared in the ground state, while the purple dashed curve is that of the initially fully excited state.}
	\label{Fig:motivating-example}
\end{figure}

Suppose now that both qubits are prepared in the ground state, \mbox{$\rho_{\ms{AB}} = \ketbra{00}{00}$}. 
Then, it is \emph{impossible} to entangle them using Eq.~\eqref{eq:01_to_Bell}, or any other energy-preserving unitary acting on the composite $\mathsf{AB}$, as $\rho_{\ms{AB}}$ has no support in $\mathcal{V}_{E}$. 
However, in the presence of a thermal bath, the populations of the composite system can be modified, thereby enabling the generation of entanglement.
A simple yet experimentally relevant example illustrating this effect involves a two-step protocol. 
First, assume access to a bosonic mode with a resonant Hamiltonian $H_{\ms R} =  E a_{\ms R}^{\dagger} a_{\ms R}$, thermalized with respect to an inverse temperature $\beta$.
Qubit $\ms B$ interacts with the bosonic mode for a finite time, which creates non-zero support of the composite state in $\mathcal{V}_{\ms E}$. 
Second, apply the unitary given in Eq.~\eqref{eq:01_to_Bell}. This two-step protocol transforms $\rho_{\ms{AB}}$ into
\begin{equation}
\sigma_{\ms{AB}} = f(\beta  E)\ketbra{00}{00}+[1-f(\beta  E)]\ketbra{\phi^+}{\phi^+}.
\end{equation}
The ground state population $f(\beta  E)$ depends on the interaction Hamiltonian and the duration of the interaction, in addition to the temperature and the energy gap. 
Let us assume that we apply Jaynes-Cummings interaction in Eq.~\eqref{eq:JC_Hamiltonian} and optimize the time duration to minimize $f(\beta  E)$.
It turns out that this optimized quantity decreases monotonically as $\beta  E$ decreases.
Hence, by additionally allowing the initial state \mbox{$\rho_{\ms{AB}} = \ketbra{00}{00}$} to interact with a thermal bath, one can generate entangled states; higher amounts of entanglement can be observed at higher temperatures as well as with smaller energy gaps. The same protocol can be applied when both spins are excited. 
However, for this initial state, contrary to the previous example, $f(\beta  E)$ is the smallest -- and thus the final state is more entangled -- when the ambient temperature is lower. In Fig.~\ref{Fig:motivating-example}, we quantify the optimal amount of entanglement produced in these two cases, measured via the negativity, as a function of the inverse temperature of the bosonic mode.

In the above discussion, we have constructed specific thermal processes that precondition the composite state into entanglable ones, by using a simple bosonic mode and the Jaynes-Cummings interaction. A natural question follows: \emph{what is the ultimate bound of the entanglement generation, for the most generic, all-encompassing set of thermal processes?}

%%%%%%%%%%%% SETTING THE SCENE %%%%%%%%%%%%

\section{Setting the scene} \label{Sec:setting-the-scene}

We investigate entanglement generation in an $N$-partite system whose initial state and the non-interacting Hamiltonian are given as
\begin{equation}
\rho \in \mathrm{conv}\qty( \bigotimes_{i=1}^N\rho_i)  \quad \text{and} \quad H = \sum_{i=1}^N \mathbbm{1}^{\otimes (i-1)}\otimes H_i\otimes \mathbbm{1}^{\otimes (N-i)}.
\end{equation}
The assumption that the Hamiltonian is non-interacting is natural, and excludes the trivial scenario of entangling simply via a fully thermalizing process.
The Hilbert space of the composite system is denoted by $\mathcal{H} = \bigotimes_{i=1}^N\mathcal{H}_i$. 
The energy level structure creates a natural division of the full Hilbert space into energy-degenerate subspaces $\mathcal{V}_E$,
\begin{equation}
    \mathcal{H} = \bigoplus_{E}\mathcal{V}_{E},
\end{equation}
where $ \mathcal{V}_E = \mathrm{span}\left\{\ket{\psi}\in\mathcal{H}:H\ket{\psi} = E\ket{\psi}\right\}$.
Note that energy-preserving unitaries acting on $\mathcal{H}$ are those that can be written as $U = \bigoplus_{E}U_{E}$, where each unitary $U_{E}$ acts on $\mathcal{V}_{E}$.
In other words, they do not mix states from different energy subspaces -- an essential feature for subsequent considerations.

\subsection{Thermal operations}

The multipartite system interacts with the thermal environment that is initialized in a Gibbs state
\begin{equation}
    \label{Eq:thermal_state}
    \gamma_{\ms R} = \frac{e^{-\beta H_{\ms R}}}{\Tr(e^{-\beta H_{\ms R}})},
\end{equation}
with a Hamiltonian $H_{\ms R}$ and an inverse temperature $\beta$.
We make minimal assumptions on the joint system-bath dynamics: the composite system is closed and thus evolves unitarily; in addition, this unitary evolution is energy-preserving. Formally, this leads to a set of operations known as \emph{thermal operations} (TOs)~\cite{Janzing2000}, which transform the state $\rho$ as follows:
\begin{equation}
    \label{Eq:thermal_operations}
\mathcal{E} (\rho)=\Tr_{\ms R} \left[U\left(\rho \otimes\gamma_{\ms R}\right)U^{\dagger}\right],
\end{equation}
where $U$ is a joint unitary that commutes with the total Hamiltonian of the system and the bath
\begin{equation}
    [U, H\otimes \iden_{\ms R}+ \iden \otimes H_{\ms R}] = 0,
\end{equation}
and the environment Hamiltonian $H_{\ms R}$ can be chosen arbitrarily. Given that there are no further constraints on $U$, the system and the bath can develop arbitrarily strong correlations (including entanglement).

\subsection{State transformation conditions }

We focus on \emph{energy-incoherent} states $\rho$, which means that each partition is (block) diagonal in the energy eigenbasis, resulting in an output state that is also (block) diagonal. Mathematically, the multipartite state $\rho$ commutes with the Hamiltonian, i.e., $[\rho,H]=0$. Under this assumption, the state $\rho$ can be fully represented by a probability vector $\v p$ consisting of the eigenvalues of $\rho$, which correspond to populations in the energy eigenbasis, $\rho = \sum_i p_i\op{E_i}$. We henceforth directly refer to $\v p= (p_1, ..., p_d)$ as energy-incoherent states, which live within the probability simplex
\begin{equation}\label{Eq:probability-simplex}
    \Delta_d= \qty{\v{p} = (p_1, ..., p_d) \in \mathbbm{R}_{\geq0}^d : \sum_i p_i = 1} .  
\end{equation}
Analogously, the thermal equilibrium state of the system, Eq.~\eqref{Eq:thermal_state} with $H_{\ms R}$ replaced by $H$, can be represented by a vector of thermal populations, denoted as $\v{\gamma}$. 

When considering transformations between different energy-incoherent states under thermal operations, an important aspect of a state is its associated $\beta$-order -- a permutation which arranges populations, rescaled with respect to the Gibbs state, in a non-increasing order.
\begin{defn}[$\beta$-ordering]\label{Def-beta-ordering}
Suppose that population vectors $\v p$ and $\v \gamma$, corresponding to a $d$-dimensional system state and its Gibbs state with respect to the ambient temperature $1/\beta$, are given. The $\beta$-ordering of $\v p$ is given as a permutation $\pi_{\v p} \in \mathcal{S}_{d}$ of energy levels within the symmetric group $\mathcal{S}_d$, such that the element-wise ratio between $\v p$ and $\v \gamma$ is non-increasing, i.e.
\begin{align}
\frac{p_{(\pi_{\v p})i}}{\gamma_{(\pi_{\v p})i}} \geq \frac{p_{(\pi_{\v p})j}}{\gamma_{(\pi_{\v p})_j}},\quad \text{for all}\quad i \leq j.
\end{align}
\end{defn}
For simplicity, we also denote a $\beta$-ordering $\pi$ as a vector $\v{\pi} = (\pi(1), \pi(2), \cdots, \pi(d))$ in the rest of the manuscript. For the specific case of $d = 4$, we denote the $\beta$-ordering $(2,1,3,4)$ as $\v{\pi}^{\star}$, since it plays a special role in entanglement generation, which will become clear in Section~\ref{Sec:entanglability} (Theorem~\ref{lem:criterion_TE}).

Finally, the set of states $\mathcal{T}_+(\v p)$ that can be achieved via thermal operations from a given initial state $\v{p}$ is called the \emph{future thermal cone}~\cite{Korzekwa2017,deoliveirajunior2022}. This set is a convex polytope comprising at most $d!$ extreme points, each corresponding to different $\beta$-orderings $\v \pi$~\cite{Korzekwa2017,deoliveirajunior2022}. Mathematically, it is represented by:
\begin{equation}
\mathcal{T}_{+}(\v p) = \operatorname{conv}[{\v p^{\v \pi}, \v \pi\in\S_d}] , 
\end{equation}
where $\v p^{\v \pi}$ are extreme points characterized by thermomajorization~\cite{Lostaglio2018elementarythermal} (see Lemma~\ref{lem_extreme}, Appendix~\ref{App:thermomajorisation}).
Again, for $d = 4$, we denote the extreme point corresponding to the ordering $\v{\pi}^{\star}$ as $\v{p}^{\star}$.

Although not explicitly stated, it is important to mention that the existence of a thermal operation that induces the transition between incoherent states $ \v p \rightarrow \v p'$ is equivalent to the existence of a Gibbs-preserving (GP) stochastic matrix mapping between these states~\cite{Janzing2000}.
Therefore, when studying the thermodynamic interconversion problem between block diagonal states, one can focus on stochastic matrices and probability vectors instead of CPTP maps and density matrices. Nonetheless, our forthcoming discussion will primarily center on the concept of thermal cones, rather than on GP matrices or thermal operations.

%%%%%%%%%%%% ENTANGLABILITY UNDER TO %%%%%%%%%%%% 

\section{Entanglability under thermal operations}\label{Sec:entanglability}

In this section, we introduce the main definitions that will enable us to derive the necessary and sufficient conditions for entangling a separable state under thermal operations. At first glance, this might appear abstract due to its generality, but these notions will be applied to a specific case in Section~\ref{Sec:ent-criteria-2-qubit}.

We begin by defining a subset of unitary operations that preserve the energy of the system and keep the energy level structure intact. Note that this is motivated by capturing the thermodynamic constraints in producing entanglement when the system is closed, i.e. when the unitary dynamics satisfy energy conservation. We refer to this subset as \emph{energy subspace-unitaries}, denoted by $\mathcal{U}_{\mathrm{subs}}$. Formally, this set is given by
\begin{equation}
    \mathcal{U}_{\mathrm{subs}} = \qty{U_{\mathrm{subs}} = \bigoplus_E U_{E}},
\end{equation}
where $U_{E}$ acts on $\mathcal{V}_{E}$ subspace of definite energy $E$. Consequently, any unitary operation from this set is an admissible operation within the framework of thermal operations. 

Next, we introduce an entanglement measure $\mathcal{M}$, which is non-negative for all separable states, $\mathcal{M}(\rho_{\ms{sep}}) \geq 0$. One such entanglement measure which we focus on is derived from the famous Horodecki-Peres positive partial transpose (PPT) criterion~\cite{Peres1996,Horodecki1996}, which states that separable states have positive partial transpose, $\rho^{T_{\ms A}}\geq0$; using this, we define $\mathcal{M}(\rho)\coloneqq\min_{\ket{\psi}}\ev{\rho^{T_A}}{\psi} = \lambda_{-}$ as the minimal eigenvalue of the partial transpose. 
Therefore, the quantity
\begin{equation}
    f(\v{p})\coloneqq\min_{U\in\mathcal{U}_{\mathrm{subs}}}\mathcal{M} \qty(U\rho U^{\dagger})
\end{equation}
is a good witness for the thermodynamic capacity of creating entanglement from $\v{p}$ via closed, energy-conserving system dynamics. 

\begin{defn}[Subspace (non-)entanglable set]\label{def:SNS}
    Given a state $\v{p}\in\Delta_d$, the subspace non-entanglable set is defined as the set of states that remain PPT under subspace rotations, i.e.
    \begin{equation}\label{eq:NE_def}
        \mathbbm{NE} \coloneqq \qty{\v{p} \, \big| \, f(\v{p})\geq 0}.
    \end{equation}
The function $f$ acts as an entanglement witness, i.e. $f(\v p)\leq 0$ implies the possibility of creating entanglement. Furthermore, the \emph{subspace entanglable set} is defined as the complement of $\mathbbm{NE}$, i.e. $ \mathbbm{E} \coloneqq \Delta_d \setminus \mathbbm{NE} $.
\end{defn}
This definition can be extended by employing multiple entanglement measures $\mathcal{M}_1,\hdots,\mathcal{M}_M$, resulting in multiple witnesses $f_1,\hdots,f_M$, and requiring that all these functions yield positive values. 

It is important to note that, states in the set $\mathbbm{E}$ are entanglable even when neither heat or work exchange are allowed.
Notably, when thermal resources are available, i.e. under thermal operations, the set of entanglable states expand. 

\begin{defn}[Future thermal cone of entanglement]\label{Def:thermal-entangled-cones}
    The future thermal cone of entanglement of a state $\v{p}$ is defined as 
    \begin{equation}
        \mathcal{T}_+^{\mathrm{ENT}}(\v p) \coloneqq \qty{\v{q} \, \big| \, f(\v{q})<0,\ \v{q}\in \mathcal{T}_{+}(\v p)} = \mathcal{T}_{+}(\v p)\cap\mathbbm{E}.
    \end{equation}
\end{defn}

Now we introduce the primary object of this paper: the set of states that cannot become entangled under thermal operations.
\begin{defn}[Thermally non-entanglable set]\label{def:TNE}
    Given the ambient temperature $\beta^{-1}$, the thermally non-entanglable set is defined as 
    \begin{equation}\label{eq:TNE_def}
        \mathbbm{TNE}(\beta) = \qty{\v{p}\, \big| \,  \min_{\v{q}\in \mathcal{T}_{+}(\v{p})}f(\v{q}) \geq 0},
    \end{equation}
with the entanglement measure $\mathcal{M}$.
Analogous to $\mathbbm{E}$, we also define the set of thermally entanglable states $\mathbbm{TE}(\beta) \coloneqq \Delta_{d}\setminus\mathbbm{TNE}(\beta)$.
\end{defn}
By definition, $\v{p}$ is thermally non-entanglable if and only if its future cone of entanglement is empty, i.e.
\begin{equation}
    \v{p} \notin \mathbbm{TNE}(\beta) ~~\Longleftrightarrow ~~\mathcal{T}_+^{\mathrm{ENT}}(\v p) = \varnothing.
\end{equation}

We can then characterize $\mathbbm{TNE}(\beta)$ with the following proposition.
Hereafter, we drop the explicit dependence on $\beta$ in $\mathbbm{TNE}$ and $\mathbbm{TE}$ when there is no concern for confusion.
\begin{prop}[Thermally non-entanglable set]\label{Thm:thermally-non-entanglable-set}
    Suppose that the entanglement measure $\mathcal{M}$ is a concave function. Then the set $\mathbbm{TNE}$ can be written as
    \begin{equation}
       \mathbbm{TNE} = \qty{\v{p}\mid \forall_\pi f(\v{p}_\pi) \geq 0}.
    \end{equation}
    When $\beta E = 0$, the above expression simplifies to
    \begin{equation}
        \mathbbm{TNE} = \bigcap_{\pi\in\mathcal{S}_d}\qty{\v{p}\mid f(\Pi\v{p}) \geq 0} =\vcentcolon \bigcap_{\pi\in\mathcal{S}_d} \,\Omega_\pi\ ,
    \end{equation}
    where $\Pi$ represents a permutation matrix corresponding to $\pi \in \mathcal{S}_d$; in particular, for $\pi = \mathrm{id}$ we have $\Omega_{\mathrm{id}} = \mathbbm{NE}$.
\end{prop}

Definition~\ref{Def:thermal-entangled-cones} also allows us to quantify the entanglement generation power using geometric measures. 
The volume of the future thermal cone of entanglement can gauge how useful an initial state is for producing entangled states. Given a state $\v p$, this quantity is defined as 
\begin{equation}\label{Eq:volume-future-ent-cone}
    \mathbf{V}^{\mathrm{ENT}}_{\mathrm{TO}}(\v p) \coloneqq \frac{V\left[\mathcal{T}^{\mathrm{ENT}}_+ (\v p)\right]}{V(\Delta_d)},
\end{equation}
where $V$ is the volume measured using the Euclidean metric. Likewise, the volume of (non-)entanglable sets, defined in Definitions~\ref{def:SNS} and~\ref{def:TNE}, are denoted as
\begin{equation}
\mathbf{V}_{\mathbbm{X}} \coloneqq \frac{V({\mathbbm{X}})}{V(\Delta_d)} \quad , \quad \mathbbm{X} = \mathbbm{NE},\ \mathbbm{E},\ \mathbbm{TNE}.
\end{equation}

\section{Two-qubit systems}\label{Sec:ent-criteria-2-qubit}

Let us now apply the general framework we developed to a specific scenario: two-qubit entanglement.
In this section, we find an efficient criterion for determining whether a two-qubit state $\v{p}$ is thermally non-entanglable, which only checks if $\v{p}^{\star}$ is subspace non-entanglable. 
We also study the geometry of the future thermal cone of entanglement and the thermally non-entanglable set for different ambient temperatures.

\subsection{Subspace entanglable set for two qubits}
Assume that Alice and Bob share a pair of qubits that are initially prepared in a separable state $\rho_{\ms{AB}}$.
Each of them is described by a local Hamiltonian $H_{A} = H_{B} =  E \ketbra{1}{1}$, and they are allowed to interact with a single heat bath via thermal operations. If there exists a thermal operation that maps $\rho_{\ms{AB}}$ to $\sigma_{\ms{AB}}$, where $[\rho_{\ms{AB}}, H] =0$ and $[\sigma_{\ms{AB}}, H]=0$, then 
 $\sigma_{\ms{AB}}$ can be expressed as
\begin{equation}\label{Eq:sigma-block-diagonal}
    \sigma_{\ms{AB}} = q_1 \ketbra{00}{00} + q_2 \ketbra{\psi}{\psi} + q_3 \ketbra{\psi_{\perp}}{\psi_{\perp}} + q_4 \ketbra{11}{11},
\end{equation}
where $\ket{\psi}, \lvert\psi^{\perp}\rangle \in \mathcal{V}_{E}$ and $\langle\psi\vert\psi^{\perp}\rangle = 0$. This observation follows from the fact that one can freely rotate the populations within the degenerate energy subspace. The unitary that connects $ \ket{\psi} = U\ket{10} $ and \mbox{$ \ket{\psi_{\perp}} = U\ket{01}$} can be written with two parameters $ \theta,\phi\in[0,2\pi]$
\begin{equation}\label{Eq:unitary}
		U = \begin{pmatrix}
			\cos\theta  & \sin\theta e^{i\phi} \\ -\sin\theta e^{-i\phi} & \cos\theta
		\end{pmatrix} \oplus\mathbbm{1}_{\mathcal{H}\setminus\mathcal{V}_{E}},
\end{equation} 
apart from an irrelevant phase factor, global to $\mathcal{V}_{E}$.
Then the final state $\sigma $ reads
\begin{equation} 
 \label{Eq:target-state}
\!\!\sigma = \begin{pmatrix}
			q_{1} & 0 & 0 & 0 \\ 0 & q_{2}\cos^{2}\theta+q_{3}\sin^{2}\theta & \frac{(q_{3}-q_{2})}{2}\sin2\theta e^{i\phi} & 0\\ 0 & \frac{(q_{3}-q_{2})}{2}\sin2\theta e^{-i\phi} & q_{2}\sin^{2}\theta + q_{3}\cos^{2}\theta & 0 \\ 0 & 0 & 0 & q_{4}
		\end{pmatrix}.
\end{equation}

For a two-qubit system, separability can be fully determined by the positive partial transposition (PPT) criterion~\cite{Peres1996,Horodecki1996}. When the partial transposition is applied to the state $\sigma$ in Eq.~\eqref{Eq:target-state}, its smallest eigenvalue is 
	\begin{equation}\label{Eq:min-eigenvalues}
	\!\!\!\lambda_{-} = \frac{1}{2}\left[\left(q_{1}+q_{4}\right) - \sqrt{\left[(q_{2}-q_{3})^{2}\sin^22\theta + (q_{1}-q_{4})^{2}\right]}\right],
	\end{equation}	
which is minimized when $ \sin^22\theta = 1$\footnote{This minimum value coincides with the case where we diagonalize the energy degenerate subspace in Bell basis, i.e., $ \ket{\psi},\lvert\psi_{\perp}\rangle = \lvert\Psi^{\pm}\rangle$.}. Thus, $\lambda_{-} $ can become negative if and only if 
    \begin{equation}\label{Eq:2qubit_PPT_pop}
		4q_{1}q_{4} < (q_{2}-q_{3})^{2}.
	\end{equation}
In other words, a state with population vector $\v{q} = (q_{1},q_{2},q_{3},q_{4})$ is subspace entanglable if and only if Eq.~\eqref{Eq:2qubit_PPT_pop} is true. We will refer to this inequality as the entanglement constraint or criterion. In addition, we define an associated entanglement measure
\begin{equation}\label{Eq:entanglement-measure}
    f(\v q) := 4q_1 q_4 - (q_2-q_3)^2,
\end{equation}
which detects entanglement whenever $f(\v q) < 0$. 

 \begin{figure}[t]
    \centering
    \includegraphics{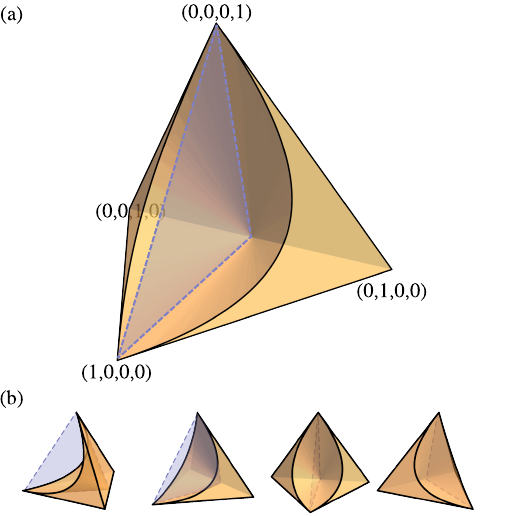}
    \caption{\textbf{Subspace entanglable set $\mathbbm{E}$}. Geometry of the set of states that can get entangled via energy-preserving unitaries acting on the degenerate energy subspace.
					Note that the set $\mathbbm{E}$ consists of two disjoints sets each on the opposite side of the plane $p_{2} = p_{3}$ coloured in blue.}
    \label{Fig-SES-set}
\end{figure}

\begin{figure*}[ht]
	\centering
	\includegraphics[width=0.85\textwidth]{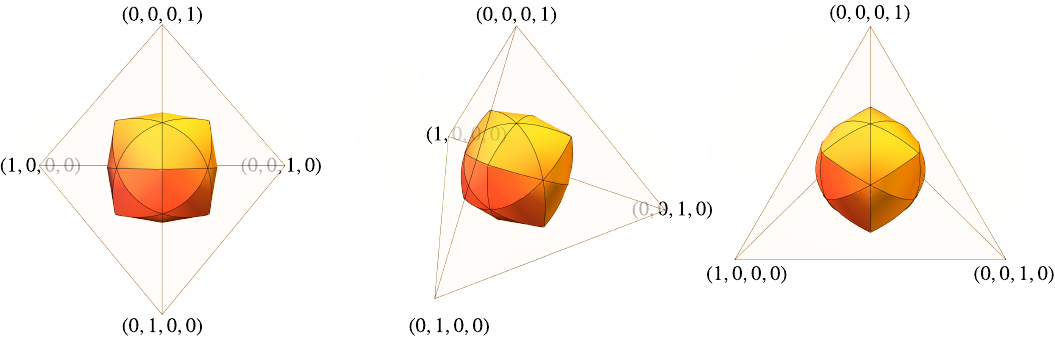}
	\caption{\textbf{Thermally non-entanglable states $\mathbbm{TNE}$ for $\beta E= 0$}. States that belong to the coloured geometrical shape cannot get entangled under thermal operations when $\beta E= 0$.  }
	\label{Fig:no-go-entanglement}
\end{figure*}

Figure~\hyperref[Fig-SES-set]{\ref{Fig-SES-set}} depicts the subspace entanglable set $\mathbbm{E}$~(see Figure~\ref{Fig-non-entanglable-set} in Appendix~\ref{App:convexity-nonentaglable} for its complement $\mathbbm{NE}$). States in $\mathbbm{E}$ can become entangled by some energy-preserving unitary of the form Eq.~\eqref{Eq:unitary}. 

The sets $\mathbbm{NE}$ and $\mathbbm{E}$ are independent of the ambient temperature and their volume can be calculated by numerical methods or direct integration. Remarkably, the normalised volume of $\mathbbm{E}$ is approximately $\mathbf{V}_{\mathbbm{E}}\approx 21/62$; much less than half of the states in $\Delta_4$ are subspace entanglable. In contrast, as we note in Section~\hyperref[Sec:quantifying]{V-C}, the future thermal cone of entanglement of the highest excited state $\v p = (0,0,0,1)$ is precisely $\mathbbm{E}$, as this state can achieve any other state in the probability simplex, regardless of the ambient temperature. Hence, the full volume $\mathbf{V}_{\mathbbm E}$ can be achieved via thermal operations. Futhermore, we prove the convexity of the subspace non-entanglable set, find its extreme points, and identify a continuous family of boundary points (Appendix~\ref{App:convexity-nonentaglable}).

\subsection{Thermally entanglable states}\label{Sec:future-thermal-entanglement}

We now construct the set of states that can generate entanglement with the help of a thermal environment. In the simple case of $\beta E = 0$, its complement, the set of thermally non-entanglable states $\mathbbm{TNE}$, can be immediately constructed from Proposition~\ref{Thm:thermally-non-entanglable-set}. The resulting set is convex and occupies approximately one third of the probability simplex [see Fig.~\eqref{Fig:no-go-entanglement}]. However, the structure of this set becomes more intricate when $\beta E\neq0$. Nevertheless, we simplify the problem by showing that the task of determining whether $\v{p}\in\mathbbm{TNE}$ -- defined by Equation~\eqref{eq:TNE_def} as including a non-trivial minimization -- is reduced to verifying if $\v{p}^{\star}\in\mathbbm{NE}$, which only requires checking a single inequality $f(\v{p}^{\star})\geq 0$. 

\begin{thm}[Criterion for determining $\v{p}\in\mathbbm{TE}$]\label{lem:criterion_TE}
	Let $\v{p}$ be a two-qubit state in $\mathbbm{TE}$. Then $\v{p}^{\star}$, the extreme point of $\mathcal{T}_{+}(\v{p})$ corresponding to the $\beta$-ordering $\v{\pi}^{\star} = (2,1,3,4)$, is in $\mathbbm{E}$.
\end{thm}

The full proof is detailed in Appendix~\ref{App:lemma_proof}, and consists in two main steps.
First, we identify $\v{\pi}^{\star}$ as a special $\beta$-ordering for subspace entanglability: whenever $\v{p}\in\mathbbm{TNE}$, there exists a state $\v{q}\in\mathcal{T}_{+}^{\mathrm{ENT}}(\v{p})$ with the ordering $\v{\pi}^{\star}$ (Lemma~\ref{Lem:extreme-beta-ordering}). Next, we use a geometric argument (Proposition~\ref{Lem:extreme_conv_points}) to prove that whenever there exists such $\v{q}$, the extreme point $\v{p}^{\star}$ is also subspace entanglable, i.e. $\v{p}^{\star}\in\mathcal{T}_{+}^{\mathrm{ENT}}(\v{p})$. It is however important to note that Theorem~\ref{lem:criterion_TE} does not imply that the minimum of $f(\v{q})$ over $\v{q}\in \mathcal{T}_{+}(\v{p})$ is always achieved when $\v{q} = \v{p}^{\star} $. For a simple counter-example, consider a state $\v{p} = (\epsilon, 1-\epsilon-\epsilon^{2},0,\epsilon^2)$ with $\epsilon\ll e^{-2\beta  E}$.
The state $\v{p}$ has a $\beta$-ordering $(2,1,4,3)$, yet a smaller $f$ value $f(\v{p})<f(\v{p}^{\star})$ compared to $\v{p}^{\star}$.

\begin{figure}[t]
	\centering
	\includegraphics{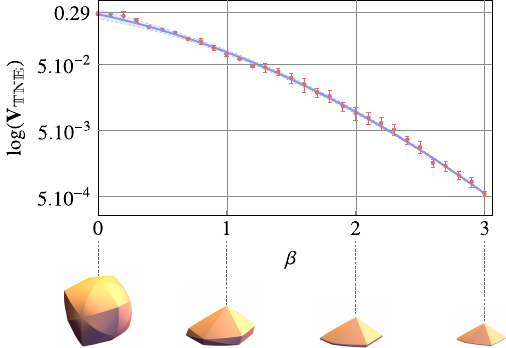}
	\caption{\textbf{Optimal entanglement generation regime.} The upper plot illustrates the approximated log volume of the set of states that cannot become entangled, as a function of  $\beta$. The bottom plot demonstrates how the geometry of this set changes at specific values of $\beta$.}
	\label{Fig:F-minimal-sep-vol}
\end{figure}

The interplay between the temperature $\beta^{-1}$ of the immediate environment and the ability to entangle a system manifests in Theorem~\ref{lem:criterion_TE} through the $\beta$-ordering.
We numerically quantify the volume of the thermally non-entanglable state set $\mathbbm{TNE}$ in Fig.~\ref{Fig:F-minimal-sep-vol}.
For any finite temperature regime $\beta > 0$, determining all states that are thermally non-entanglable is not straightforward.
We devise a computationally efficient algorithm, using Theorem~\ref{lem:criterion_TE} and the vacuum-packing-inspired procedure, that approximates $\mathbbm{TNE}$ for given $\beta$:
\begin{enumerate}[leftmargin=1.5em]\itemsep0em 
	\item Generate a regular grid on the surface of the probability simplex $\Delta_3$.
	\item  Select a point $\v{p}^{(0)}_o$ from the grid and the Gibbs state $\v{p}^{(0)}_i=\boldsymbol{\gamma}$. With this choice, we are guaranteed to have $\v{p}^{(0)}_o\in\mathbbm{TE}$ and $\v{p}^{(0)}_i\in\mathbbm{TNE}$\footnote{Except for a single case $\v{p}^{(0)}_{o} = \frac{1}{1+2e^{-\beta E}}(1,e^{-\beta E},e^{-\beta E},0)$.}.
	\item In the $j$-th step of the procedure we define midpoint $\v{p}^{(j)}_m = \frac{1}{2}\qty(\v{p}^{(j)}_i + \v{p}^{(j)}_o)$ and, based on Theorem~\ref{lem:criterion_TE}, decide whether $\v{p}^{(j)}_m$ is an element from the non-entanglable set. If $\v{p}^{(j)}_m\in\mathbbm{TNE}$, then $\v{p}^{(j+1)}_i = \v{p}^{(j)}_m$; else $\v{p}^{(j+1)}_o = \v{p}^{(j)}_m$.
	\item Repeat the bisection step $n=30$ times, achieving the precision $\norm{\v{p}^{(n)}_i - \v{p}^{(n)}_o} \approx 10^{-9}\norm{\v{p}^{(0)}_i - \v{p}^{(0)}_o}$.
	\item Add the point $\v{p}^{(n)}_m$ to the approximate set $\widetilde{\mathbbm{TNE}}$.
 \item Repeat steps 2-5 for all the gridpoints.
\end{enumerate}
The convex hull of the final set $\operatorname{conv}(\widetilde{\mathbbm{TNE}})$ is a good approximation of the actual non-entanglable set $\mathbbm{TNE}$.

The numerical result in Fig.~\ref{Fig:F-minimal-sep-vol} exhibits a monotonic behaviour, indicating that low temperatures enhance entanglability. 
Indeed, this outcome is intuitively expected.
Thermal processes permit the local bypassing of strict energy-preservation, thereby offering advantage in terms of entanglement generation.
Yet, thermal noise inevitably accompanies thermal processes, impacting the amount of achievable entanglement. 
The negative impact of thermal noise is anticipated to be minimized at lower environmental temperatures.
We confirm this first by observing that the set of non-entanglable states $\mathbbm{TNE}$ vanishes at the zero temperature limit $\beta\rightarrow\infty$. Moreover, we can make a more rigorous observation. 
\begin{lem}
    The set $\mathbbm{TNE}(\infty) = \{(1,0,0,0)\}$, i.e. it contains only the ground state. 
\end{lem}
\begin{proof}
Given any other initial state, it is possible to thermalize one of the subsytems; in the zero-temperature limit ($\beta\rightarrow
\infty$), this transforms a state $\v{p} = (p_{1},p_{2},p_{3},p_{4})$ into $\v{q} = (p_{1}+p_{3},p_{2}+p_{4},0,0)$ or $(p_{1}+p_{2},p_{3}+p_{4},0,0)$. However, in this specific scenario, it is sufficient to assume entanglement because the last two levels are empty while the second level is populated. Consequently, $f$ always becomes negative. Thus, for any state $\v{p} \neq (1,0,0,0)$, the final state $\v{q}$ is subspace entanglable, i.e. $\v{p}\in\mathbbm{TE}$.    
\end{proof}

One may be tempted to view the ground state $(1,0,0,0)$ then as useless for entanglement generation; however, this is not true. In particular, the ground state becomes thermally entanglable whenever $\beta^{-1}\neq0$: one can thermalize one of the subsystems to obtain $\v{q} = \frac{1}{1+e^{-\beta E}}(1,e^{-\beta E},0,0)$, which is subspace entanglable.
What is observed is that as long as there is a certain amount of athermality in the initial state (with respect to environment temperature), we can often trade it for some amount of entanglement. 

Given the above, it is then intuitive to see that given an initial state $\v{p}$, the range of environmental temperature that allows for entanglement generation is very much state-dependent. In particular, we are able to identify critical inverse temperatures $\beta_{C}$, at which a state $\v{p}$ starts to become thermally (non-)entanglable. In fact, we show that $\beta_{C}$, if exists, can be directly calculated for any state by employing Theorem~\ref{lem:criterion_TE} alongside the entanglement constraint in Eq.~\eqref{Eq:2qubit_PPT_pop}. 
Interestingly, for some states two critical temeperatures $\beta_{C_{1}}<\beta_{C_{2}}$ exist and the state becomes entanglable when $\beta<\beta_{C_{1}}$ or $\beta>\beta_{C_{2}}$.

A particularly clear analysis can be done when the initial state themselves are thermal, having a temperature $\beta_{\ms{S}}\neq\beta$ in principle different from the ambient temperature. We report the conclusions here, while directing the curious reader to the full calculation in Appendix~\ref{A:thermal-state}.

Thermal initial states are written as
\begin{equation}
    \v p = \frac{1}{Z_{\ms{S}}}(1, e^{-\beta_{\ms{S}}  E},e^{-\beta_{\ms{S}}  E}, e^{-2\beta_{\ms{S}}  E}),
\end{equation}
where $Z_{\ms{S}} = (1+e^{-\beta_{\ms{S}} E})^{2}$. 
When the system is colder than the environment, $\beta_{\ms S} > \beta$, the $\beta$-ordering of the state is $(1,2,3,4)$.
Finding the explicit form of the extreme point $\v{p}^{\star}$ and solving $f(\v{p}^{\star}) = 0$, we may obtain the critical temperature 
\begin{align}\label{eq:betaC1_explicit}
	\beta_{C_{1}}  \equiv \beta_{\ms S} -E^{-1} \log\left[ 1 + 2e^{-\beta_{\ms S} E}\left(\sqrt{e^{2\beta_{\ms S} E}+1} - 1\right)\right]\ .
\end{align}
Whenever the ambient temperature $\beta^{-1}$ is higher than the critical temperature $\beta_{C_{1}}^{-1}$, the thermal state $\v{p}$ becomes thermally entanglable. 
We can obtain a simple approximation $\beta_{C_{1}}  \simeq \beta_{\ms S} -E^{-1}\log3$ of Eq.~\eqref{eq:betaC1_explicit} up to the zeroth order of $e^{-\beta_{\ms S} E}$ in the low system temperature limit $\beta_{\ms{S}} E\gg 1$.

On the other hand, when the system is hotter than its surrounding, i.e. $\beta_{\ms S} < \beta$, the critical temperature can be obtained from a quartic equation and we do not have an explicit formula for it. 
However, in the same limit  $\beta_{\ms{S}} E\gg 1$, we can approximate $\beta_{C_{2}} \simeq \beta_{\ms S} +E^{-1}\log3$.
In other words, if 
\begin{equation}
    |\beta_S - \beta| \geq E^{-1}\log3,
\end{equation}
then we may operate such a thermal machine to generate entanglement. See Fig.~\ref{fig:thermal_to_entangled} for the numerical values of the critical temperatures, which clearly displays the validity of our approximation $\beta_{C} \simeq \beta_{\ms{S}} \pm E^{-1}\log3$.

As mentioned earlier, the temperature difference between the system and the environment is a resource for entanglement generation, as it is for many other thermodynamic tasks. This process can be conceptualized as a type of heat engine which aims not to output work, but entanglement. For instance, consider a 3-stroke thermal-entanglement engine whose working substance is a two-qubit system. This engine operates between an ambient bath at temperature $\beta^{-1}$ and a working bath with temperature $\beta^{-1}_W:\abs{\beta^{-1}_W - \beta^{-1}_C}>0$. The cycle proceeds as follows: \emph{(i)} the two-qubit system thermalises with the working bath, reaching equilibrium at the working bath's temperature $\beta_W^{-1}$; \emph{(ii)} the system is then brought into contact with the ambient bath, which is at temperature $\beta^{-1}$. During this stage, entanglement is generated within the two-qubit system due to thermal interactions with the ambient bath. To close the cycle, \emph{(iii)} the generated entanglement is used, and the system is returned to its initial state, which is energy-incoherent.

\begin{figure}[t]
	\includegraphics{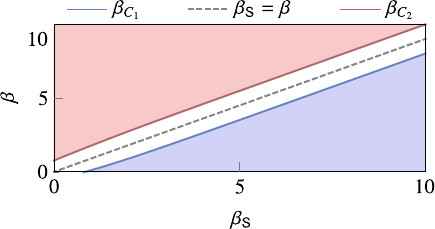}
	\caption{\textbf{Thermal state \& entanglability}. Critical temperatures of an environment for thermal initial states, where the qubit energy gap is set to $E=1$. 
		The $x$-axis indicates the inverse temperature $\beta_{\ms S}$ of the initial state and the $y$-axis marks the ambient inverse temperature $\beta$.
		If the pair $(\beta_{\ms S},\beta)$ is in either blue or red shaded region, the thermal initial state is entanglable. 
		The blue and red lines mark the hot critical (inverse) temperature and cold critical (inverse temperature) for the environment to make the thermal initial state entanglable. 
		}
		\label{fig:thermal_to_entangled}
\end{figure}

\subsection{Quantifiers for achievable entanglement under thermal operations}\label{Sec:quantifying}

In the previous sections, we focused on determining whether a state is entanglable or not using thermal resources. Here, we are additionally interested in the extent of entanglement that thermal operations can produce, and how this varies across different temperature regimes. 
Two figures of merit can be used to quantify this extent: the \emph{thermal volume of entanglement }, $\mathbf{V}^{\mathrm{ENT}}_{\mathrm{TO}}(\v p)$ as defined in Eq.~\eqref{Eq:volume-future-ent-cone} and the maximum negativity~\cite{Vidal2002} that the initial state can produce. In particular, the thermal volume of entanglement 
is a thermodynamic monotone, in other words, $\mathbf{V}^{\mathrm{ENT}}_{\mathrm{TO}}(\v p) \geq \mathbf{V}^{\mathrm{ENT}}_{\mathrm{TO}}(\v q)$ for all $\v{q}\in\mathcal{T}_{+}(\v{p})$. Further insights on this quantity can be developed by considering the geometry of the future thermal cone of entanglement $\mathcal{T}_{+}^{\mathrm{ENT}}(\v{p})$. Recall in Definition \ref{Def:thermal-entangled-cones}, $\mathcal{T}_{+}^{\mathrm{ENT}}(\v{p})$ is the intersection of two sets: $\mathbbm{E}$ and $\mathcal{T}_{+}(\v{p})$. We also know that $\mathbbm{E}$ is non-convex and not connected, hence $\mathcal{T}_{+}^{\mathrm{ENT}}(\v{p})$ is also non-convex and not connected [see bottom panel of Fig.~\eqref{Fig-future-evolution}].
However, as the temperature decreases, two disjoint parts grow in sizes and approach each other, becoming almost connected; this is a natural consequence of the future thermal cone approaching the plane $p_{4} = 0$, which is on the boundary of $\mathbbm{E}$.

\begin{figure}
	\centering
	\includegraphics{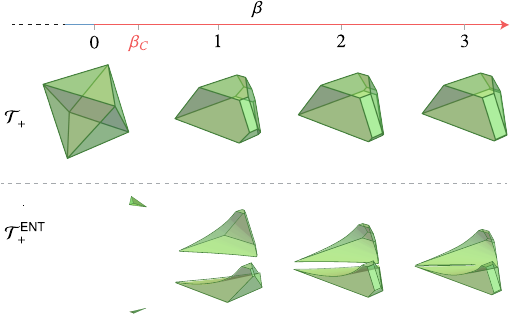}
	\caption{\textbf{Regions of entanglement}. The evolution of the future and entanglement thermal cones as a function of $\beta$ reveals interesting behavior. For a bipartite product state $\v{p} = (0.5, 0.5) \otimes (0.24, 0.76)$ with an identical energy gap $E = 1$, we observe that at $\beta = 0$, the state cannot be entangled via thermal operations. However, for $\beta \geq \beta_C \approx 0.21$, the state becomes entanglable. As $\beta$ increases, the two disjoint parts of $\mathcal{T}^{\mathrm{ENT}}(\v p)$ become closer, until at $\beta \to \infty$, they are separated only by a plane $p_{2}=p_{3}$.}
	\label{Fig-future-evolution}
\end{figure}

The change in volume $\mathbf{V}^{\mathrm{ENT}}_{\mathrm{TO}}(\v p)$ with respect to $\beta$ is, however, heavily dependent on the initial state. 
To see this, consider the energy eigenstates, i.e. four vertices of the probability simplex $\Delta_{4}$ as described explicitly in the caption of Figure \ref{Fig:volume-negativity}.
For the ground state $\v{p}_1$, $\mathbf{V}^{\mathrm{ENT}}_{\mathrm{TO}}(\v{p}_1)$ decreases with $\beta$ [black solid curve, top panel of Fig.~\ref{Fig:volume-negativity}], because the future thermal cone $\mathcal{T}_{+}(\v{p}_1)$ shrinks as the Gibbs state $\gamma$ approaches the ground state as $\beta$ increases. On the other hand, for the maximally excited state $\v{p}_4$, $\mathcal{T}_{+}(\v{p}_4)$ remains constant [red dot curve, top panel of Fig.~\ref{Fig:volume-negativity}], since the future thermal cone encompasses the entire probability simplex, regardless of the ambient temperature. In this case, the future thermal cone of entanglement is precisely $\mathbbm{E}$. 
Finally, for both $\v{p}_2$ and $\v{p}_3$, the thermal volume of entanglement decreases with $\beta$, albeit slower than that of $\v{p}_1$. 

\begin{figure}[t]
	\centering
	\includegraphics{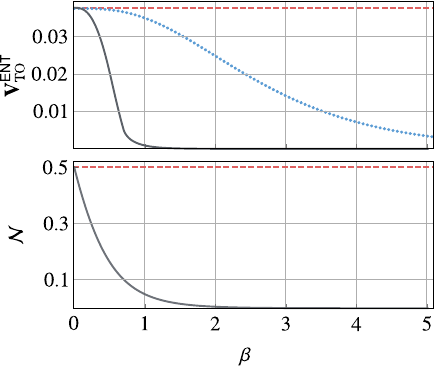}
	\caption{\textbf{Volume and negativity as a function of $\beta$.} For all sharp states $\v{p}_{1}=(1,0,0,0)$, $\v{p}_{2}=(0,1,0,0)$, $\v{p}_{3}=(0,0,1,0)$, and $\v{p}_{4}=(0,0,0,1)$, we plot the thermal volume of entanglement (top), and the largest negativity as a function of the inverse temperature. Among all states, two are distinct: when both qubits start in the excited state $\v{p}_{4}$ and when they start in the ground state (black solid curve) $\v{p}_{1}$. Any other permutation of the initial state characterizes a different active state, and their volume is represented by the blue dotted curve on the top panel. Note that this plot also represents the normalized volume of subspace entanglable states, since the thermal cone of $\v{p}_{4}$ is the full set space. Moreover, although the maximum achievable negativity of $\v{p}_{4}$ is the same as, e.g., $(0,1,0,0)$, it is much more powerful in terms of state preparation, which can be useful since entanglement is in general irreversibly manipulated.
		In the bottom panel, the maximum negativity of these states coincides with the maximally active state.}
	\label{Fig:volume-negativity}
\end{figure}

The second figure of merit is negativity~\cite{Vidal2002}, 
\begin{equation}\label{Eq:negativity}
    \mathcal{N}(\rho) = \frac{\norm{\rho^{T_{\ms A}}}_1 - 1}{2} = \sum_i \frac{|\lambda_i| - \lambda_i}{2},
\end{equation}
where $\lambda_i$ represents all the eigenvalues of $\rho^{T_{\ms A}}$.
Equivalently, $\mathcal{N}(\rho)= \sum_{j}\vert{\lambda}^{-}_{j}\vert$, where ${\lambda}^{-}_{j}$ are negative eigenvalues of $\rho^{T_{\ms A}}$.
Therefore, if the system is qubit-qubit or qubit-qutrit, non-zero negativity is necessary and sufficient for the state being entangled; for any larger systems, there exist entangled states with zero negativity~\cite{horodecki1998mixed, chen2016schmidt}.
In the qubit-qubit system we are examining, given a population vector $\v{q} = \mathrm{eig}(\sigma)$, the maximal negativity can be obtained from Eq.~\eqref{Eq:min-eigenvalues} as
\begin{equation}\label{Eq:negativity-state}
    \mathcal{N}(\sigma) =\frac{1}{2}\qty[\sqrt{(q_1-q_4)^2+(q_2-q_3)^2}-(q_1+q_4)],
\end{equation}
when $f(\v{q})<0$, which has also been expected from Ref.~\cite{Verstraete2001} in a different context.
Observe that
Eq.~\eqref{Eq:negativity-state} is related to Eq.~\eqref{Eq:entanglement-measure} via 
\begin{align}
    \mathcal{N}(\sigma) = \frac{1}{2}\qty[\sqrt{(q_1+q_4)^2-f(\v q)}-(q_1+q_4)],
\end{align}
and monotonically decreases as $f(\v q)$ increases.
For $f(\v{q})\geq0$, the negativity is always zero. 

The maximum negativity over all probability simplex $\Delta_{4}$ occurs when $\v{q} = (0,1,0,0)$ or $\v{q} = (0,0,1,0)$ and the corresponding state are the Bell states. 
If the system is initially in one of the four vertices of $\Delta_{4}$, it is possible to prepare Bell states via thermal operations when $\beta = 0$.
However, as $\beta \geq 0$, the ground state $\v{p}_1$ loses its ability to produce a Bell state, while three other pure states remain capable. 
The maximum negativity that $\v{p}_1$ can produce decreases as $\beta$ increases.

\section{Additional discussions}\label{sec:discussions}
\subsection{Generation of entanglement in the presence of dissipation}
\label{sec:dissipative}

So far, our discussion revolved around the assumption that any energy-preserving system-bath dynamics is accessible. 
However, one can also consider the scenario where the system undergoes an open dynamics described by a Lindblad master equation~\cite{kossakowski1972quantum,gorini1976completely,lindblad1976generators},
	\begin{equation}
		\label{Eq:master_equation}
		\frac{d \rho(t)}{dt}  = -i \qty[H, \rho(t)] + \mathcal{L}_t\qty[\rho(t)],
	\end{equation}
where $[\cdot,\cdot]$ denotes the commutator and $\mathcal{L}_t$ is the Lindbladian having the general form
	\begin{equation}
		\label{eq:lindbladian}
		\mathcal{L}_t(\rho) = \sum_{i} r_i(t) \left[ L_i(t) \rho L_i(t)^\dag - \frac{1}{2}\Bigl\{L_i(t)^\dag L_i(t), \rho\Bigl\}\, \!\right]\!,\!
	\end{equation} 
with $\{\cdot,\cdot\}$ denoting the anticommutator, $L_i(t)$ being jump operators, and $r_i(t)\geq 0$ being non-negative jump rates. 
Furthermore, assuming that the quantum system interacts with a large heat bath implies that the system Gibbs state is a stationary solution of the dynamics, $\mathcal L_t(\gamma) = 0$ and that the Lindbladian $\mathcal L_t$ commutes with the generator of the Hamiltonian dynamics $-i[H,\cdot]$ for all times $t$. 

 \begin{figure}[t]
    \centering
    \includegraphics{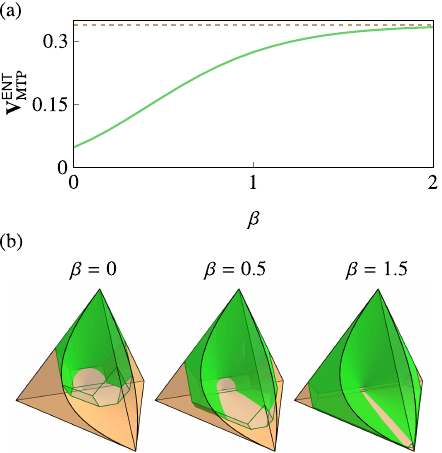}
    \caption{\textbf{Entanglement generation via Markovian thermal process}. The upper plot illustrates the volume of the future thermal cone of entanglement as a function of $\beta$ for a bipartite system $\rho_{\ms AB} = \ketbra{11}{11}$ undergoing a Markovian thermal process, depicted by the solid green curve. The dashed orange line represents the volume of the future thermal cone of entanglement for the same state, but under thermal operations. In the lower plot, the evolution of the Markovian future thermal cone of entanglement (green) with respect to $\beta$ is shown. Here, the orange shape depicts the future thermal cone under TO, while the internal lighter color with green vertices represents the future cone under MTP.}
    \label{Fig:F-MTP}
\end{figure}

The dynamics generated by master equations that satisfy two aforementioned properties lead to what is known as a Markovian thermal process, which is known to be a strict subset of thermal operations~\cite{Lostaglio22_MTO1}. The differences between them arise from our degree of control on the bath. Nonetheless, the entanglement criterion given in Eq.~\eqref{Eq:2qubit_PPT_pop} remains valid for Markovian thermal processes. Consequently, we can study the entanglement generation in the presence of dissipation using the same criterion function $f$. 

The additional restriction stemming from the Markovianity of the Lindblad master equations is the most evident when considering the highest excited state $\v{p} = (0,0,0,1)$ as the initial state. 
As analysed in Section~\hyperref[Sec:quantifying]{V-C}, the volume of the future thermal cone of entanglement $\mathbf{V}^{\mathrm{ENT}}_{\mathrm{TO}}(\v p)$ defined in Eq.~\eqref{Eq:volume-future-ent-cone} remains maximal for all inverse temperatures $\beta$, when we can utilize any thermal operation to generating entanglement.
In Fig.~\hyperref[Fig:F-MTP]{\ref{Fig:F-MTP}a}, we present how the volume $\mathbf{V}^{\mathrm{ENT}}_{\mathrm{TO}}(\v p)$ varies as a function of $\beta$, when only Markovian thermal processes are allowed.
For any $\beta$, the volume $\mathbf{V}^{\mathrm{ENT}}_{\mathrm{TO}}(\v p)> 0$ indicates that the entanglement generation via dissipative dynamics is possible. 
However, the difference between the entanglement generation via thermal operations is the most pronounced in the high-temperature limit, i.e. when $\beta$ is small.
Furthermore, in the limit $\beta\rightarrow\infty$, Markovian thermal processes are as capable as thermal operations for generating entanglement. The detailed behaviour of the future thermal cone of entanglement under MTP is depicted in Fig.~\hyperref[Fig:F-MTP]{\ref{Fig:F-MTP}b}.

\subsection{Remarks on higher-dimensional systems}
\label{sec:multiparty}

The entanglement structure for multipartite systems with more than two parties is notoriously complex. 
Multipartite entanglement can be classified using equivalence classes defined with invertible local transformations. Two states are considered to have the same class of entanglement if one can be transformed into the other via LOCC with a non-zero probability~\cite{Cirac2000}. Even in the simplest three-qubit systems, one encounters two distinct classes of entanglement: GHZ-like entanglement and W-like entanglement~\cite{Cirac2000,Sanpera2001}, represented by the states
\begin{align}
    \ket{\mathrm{GHZ}} = \frac{\ket{000} + \ket{111}}{\sqrt{2}},\ \ 
    \ket{\mathrm{W}} = \frac{\ket{001} + \ket{010} + \ket{100}}{\sqrt{3}}.
\end{align}

The $n$-qubit extensions of two classes exist. In particular, $n$-qubit GHZ states are defined as
\begin{equation}
\ket{\mathrm{GHZ}_n} = \frac{1}{\sqrt{2}}\left( \underbrace{\ket{0\hdots0}}_{n} + \underbrace{\ket{1\hdots1}}_{n} \right),
\end{equation}
an equal superpositions of all $n$ qubits being in the ground state and in the excited state. The concept of W-like entanglement is similarly generalized to $n$-qubit Dicke states
\begin{equation}
\ket{D_{m}} = \binom{n}{m}^{-1/2}\sum_{\sigma\in\mathcal{S}_{n}}P_\sigma \qty(\underbrace{\ket{0\hdots0}}_{n-m}\otimes\underbrace{\ket{1\hdots1}}_{m})\ ,
\end{equation}
where the sum goes over all permutations $\sigma$ on $n$ elements and $P_\sigma$ are the respective matricial representations.
These states form a basis for the symmetric subspace of an $n$-qubit system. 

Remarkably, different energy level structures give rise to different entanglement structures that can be achieved using thermal operations. Let us assume that the Hamiltonian of the $i$-th qubit is $H_i = E_i\ketbraa{1}{1}$. First, if all the energy gaps are equal, i.e. $E_i = E$, then the $n$-qubit Hilbert space is divided into energy-degenerate subspaces:
\begin{equation}
\mathcal{H} = \bigoplus_{i=0}^n \mathcal{V}_{i E},
\end{equation}
each of which spanned by states with a fixed number of excitations $i$ and thus the energy equal to $i\cdot E$. 
With this energy-level structure, Dicke states $\ket{D_i}$, or W-like entanglement in the three-qubit case, can be generated via thermal operations with nonzero components within the energy-degenerate subspaces.

On the other hand, if we set $E_n = \sum_{i=1}^{n-1} E_i$, we find that the degenerate subspace corresponding to this energy
\begin{equation}
    \mathcal{V}_{E_n} = \mathrm{span}\qty(
    |\underbrace{1\hdots1}_{n-1}0\rangle,\,
    |\underbrace{0\hdots0}_{n-1}1\rangle)
\end{equation}
contains states possessing GHZ-like entanglement. Specifically, in the case of three qubits, this configuration leads to GHZ states up to a NOT operation on the third qubit. Similar arguments can be easily extended to systems with larger local dimensions. 

Concerning systems with larger local dimensionality, we note that PPT criterion, based on negativity of partial transpose of a state, is a complete entanglement measure not only for two-qubit systems, but also for qubit-qutrit systems as well \cite{Horodecki1996}. In particular, we may consider a system in an energy-incoherent state $\rho$ described by population $\v p=\qty(p_1,\hdots,p_6)$ and subject to local hamiltonians $H_1 = \op{1}$ and \mbox{$H_2 = \op{1} + 2\op{2}$}. Based on similar arguments as for the qubit-qubit systems, if any of the following two conditions,
\begin{equation}
    \begin{aligned}
        4\,p_1 \min(p_4,p_5) & < (p_2-p_3)^2, \\
        4p_6 \min(p_2,p_3) & < (p_4-p_5)^2,
    \end{aligned}
\end{equation}
is satisfied, one of the eigenvalues of the partial transpose $\rho^{T_A}$ can become negative under energy-preserving unitary operations, and therefore, the state is entanglable. Similar expressions can be easily put forward for qubit-qudit systems.

It should be highlighted, however, that in general, complete entanglement measures for multipartite systems beyond $2\times2$ and $2\times3$ do not exist, unlike the one based on the PPT criterion, which is effective for two-qubit and two-qutrit  cases \cite{Peres1996, Horodecki1996}, which is reflected in the appearance of so-called bound entanglement \cite{Gaida_bound}. 
Nevertheless, there exist partial entanglement measures, such as the 3-tangle for three-qubit systems~\cite{Wootters2000}, can be used to provide an outer approximation of the entanglable and non-entanglable sets in such scenarios.

\section{Outlook} \label{Sec:outlook}

In this paper, we have investigated the interplay between entanglement and thermodynamics. 
The minimal thermodynamic assumptions we made, namely that the system-bath is closed and evolves via an energy-preserving unitary, constrain the final entanglement achievable from given initial states. By applying the PPT criterion and constructing future thermal cones, we derived necessary and sufficient conditions for two-qubit initial states that can be transformed into entangled states within thermal operations framework. 
It turns out that any separable state, apart from a cluster of states around the thermal state, can become entangled.
In other words, our investigation demonstrates that entanglement generation is possible, given certain degree of athermality in the initial state.
Furthermore, we captured how the entanglement generating capability varies with the ambient temperature, by analysing the volume of entanglable initial state set, and found that low temperature allows more states to become entanglable.
Nevertheless, when the initial state is fixed, we discovered that lowering temperature is not always beneficial and that more detailed considerations on the geometry of thermal cones are needed. 

There are several promising avenues for extending and generalizing the results presented in this paper. A concrete and direct generalization consists of using an auxiliary system, a catalyst, to facilitate a process that would not occur spontaneously; see recent reviews~\cite{datta2022catalysis, Patrykreview}. The idea of using a catalyst in our context arises from Ref.~\cite{lipkabartosik2023fundamental}, where the authors demonstrate that catalysis enables the creation of otherwise inaccessible types of correlations by accessing those locked in non-degenerate energy subspaces. As we have discussed, the entanglement generation capacity depends on the structure of degenerate energy subspaces. Employing a catalyst not only increases the number of degenerate energy subspaces but also lift some of the dynamical restrictions imposed by energy conservation. In this regard, the initial steps demonstrating the utility of generating entanglement via thermal operations with the aid of a strict catalyst were presented in Ref.~\cite{catalysiskuba2024}. Furthermore, in Appendix~\ref{App:catalytic-advantages}, we provide an example where correlated catalysis enables the generation of entanglement from an initial state that cannot be entangled without a catalyst. This raises the intriguing question of how one can systematically study the generation of entanglement with the assistance of a catalyst.

In this work, we only considered entanglement generation from initially energy-incoherent states.
When we focus on the entanglement generated between degenerate energy levels, coherence between different energy levels cannot contribute.
Nevertheless, coherence can still nontrivially affect the entanglement, including those between energy levels with different energy values.
Currently, tools and methods available for analysing coherent state transformations under thermal operations are rather limited~\cite{Korzekwa2017,LostaglioKorzekwaCoherencePRX, Gour_2018}. 
For two-qubit systems, existing techniques to characterize single qubit coherent state transformations~\cite{LostaglioKorzekwaCoherencePRX,Korzekwa2017} may be adapted to study the entanglement generation.
This direction of study will extend our results to the realm where three resources -- entanglement, athermality, and coherence -- compete and cooperate with each other.

We make a few final remarks on existing literature, to which our work can connect. 
First, it would be interesting to study the conventional quantifiers of entanglement such as entanglement cost and distillable entanglement~\cite{horodecki2009quantum} which are originally defined via LOCC, and evaluate them with regard to thermodynamic constraints instead. In particular, the gap between these quantities would then quantify the amount of irreversibility when it comes to the local preparation of entanglement.
Lastly, entanglement in continuous-variable (CV) systems is relatively less explored.  
Extending our analysis to CV systems would be extremely fruitful, considering their widespread use in practical experiments~\cite{andersen2010continuous} involving optomechanics~\cite{Rocheleau2009, Teufel2011}, and superconducting setups~\cite{PhysRevA.88.032317}.
For instance, we may use the toolkits developed in the study of Gaussian thermal operations~\cite{Serafini_2020} for an analysis similar to what have been done in the current work, and aim to explicate CV entanglement results, such as the recent no-go result~\cite{Longstaff2023} for the CV steady-state entanglement generation.

\begin{acknowledgments}

We thank Kamil Korzekwa for valuable discussions that significantly contributed to the execution of this project, and also Patryk Lipka-Bartosik for very fruitful comments on the first version of the manuscript. AOJ acknowledges financial support from VILLUM FONDEN through a research grant (40864).
JS and NN are supported by the start-up grant for Nanyang Assistant Professorship of Nanyang Technological University, Singapore, awarded to Nelly H.~Y.~Ng.
JCz acknowledges financial support by NCN PRELUDIUM BIS no. DEC-2019/35/O/ST2/01049.
\end{acknowledgments}

\bibliography{Bibliography}

\newpage
\appendix
\onecolumngrid

\section{Future thermal cone}\label{App:thermomajorisation}
\begin{figure*}
    \centering
    \includegraphics{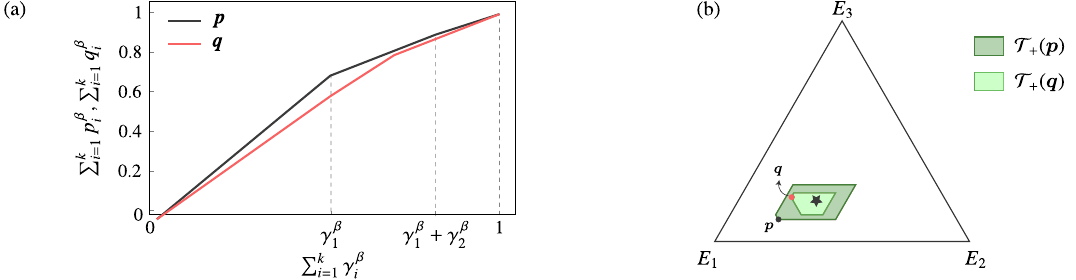}
    \caption{\textbf{Thermomajorisation $\&$ future thermal cone}. For two  states $\v p = (0.7, 0.2, 0.1)$ and $\v q = (0.6, 0.2, 0.2)$ with $\beta = 0.5$, (a) we plot their thermomajorization curves $\mathcal{L}_{\v p}(x)$ and $\mathcal{L}_{\v q}(x)$, respectively. Since $\v p$ thermomajorizes $\v q$, it implies that one can transform $\v p$ into $\v q$ under thermal operations. This can also be observed by examining their (b) future thermal cones, where $\v q$ lies within the future thermal cone of $\v p$. The Gibbs state $\v{\gamma}$ is depicted by a black star $\star$. 
    }
    \label{F:app-example}
\end{figure*}

We briefly review state transition rules under thermal operations and the notion of the future thermal cone, a set of state that can be transformed from a given initial state with thermal operations. 
First, we introduce thermomajorization relations, which fully characterize the convertibility between two states via thermal operations.
The thermomajorization curves, associated to each energy-incoherent state, are defined as the following: 
\begin{defn}[Thermomajorization curve]\label{Def:thermomajorisation-curve} 
    Suppose that the Gibbs state $\v{\gamma}\in\Delta_{d}$ with respect to the system Hamiltonian $H$ and the ambient inverse temperature $\beta$ is given. 
    For each state $\v p \in \Delta_d$, we define the thermomajorisation curve of $\v{p}$ as a piecewise-linear function $\mathcal{L}_{\v p}: [0,1] \to [0,1]$ composed of linear segments connecting the point $(0,0)$ and the elbow points $\{(\sum_{i=1}^k\gamma_{\pi_{\v{p}}^{-1}(i)},~\sum_{i=1}^k p_{\pi_{\v{p}}^{-1}(i)})\}_{k=1}^{d}$. 
    Here $\pi_{\v{p}}$ is the $\beta$-ordering of $\v{p}$, which imposes $p_{\pi_{\v p}^{-1}(i)}/\gamma_{\pi_{\v p}^{-1}(i)} \geq p_{\pi_{\v p}^{-1}(j)}/\gamma_{\pi_{\v p}^{-1}(j)}$ for all $i\leq j$.
\end{defn}
Note that thermomajorization curves are always concave and non-decreasing by definition. 
For the Gibbs state $\v{\gamma}$, the thermomajorization curve $\mathcal{L}_{\v{\gamma}}$ is a straight line from $(0,0)$ to $(1,1)$.
From the concavity of the curve and the fact that any $\mathcal{L}_{\v{p}}$ passes points $(0,0)$ and $(1,1)$, any thermomajorization curve is above the curve for the Gibbs state. i.e. $\mathcal{L}_{\v{p}}(x)\geq\mathcal{L}_{\v{\gamma}}(x)$ for all $x\in[0,1]$.
See Fig.~\ref{F:app-example} (a) for examples of thermomajorization curves. 

Using thermomajorization curves, thermomajorization relations between two states can be established.
\begin{defn}[Thermomajorization relation]
    The thermomajorization relation is a preorder between two energy-incoherent states.
    If two states $\v{p},\v{q}\in\Delta_{d}$ have corresponding thermomajorization curves $\mathcal{L}_{\v p}$ and $\mathcal{L}_{\v q}$, such that $\mathcal{L}_{\v p}(x) \geq \mathcal{L}_{\v q}(x)$ for all $x \in [0,1]$, we say $\v{p}$ thermomajorizes $\v{q}$, or more simply,  $\v p \succ_{\beta} \v q$.
\end{defn}
Since the thermomajorization relation gives a preorder, not a total order, there exist pairs of states $\v{p}$ and $\v{q}$ such that neither $\v{p}\succ_{\beta}\v{q}$ nor $\v{q}\succ_{\beta}\v{p}$ holds.
Such pairs are dubbed incomparable. 
Reamrkably, thermomajorization relations coincide with the state convertibility under thermal operations.
\begin{thm}[Theorem~2 of Ref.~\cite{horodecki2013fundamental}]\label{thm:thermomajorisation}
    Suppose that $\v{p},\v{q}\in\Delta_{d}$ describe two energy-incoherent states of a system with Hamiltonian $H$ with access to thermal environments with the inverse temperature $\beta$.
    If $\v{p}\succ_{\beta}\v{q}$, than $\v{p}$ can be transformed to $\v{q}$ via some thermal operation.
\end{thm}
In particular, the Gibbs state $\v{\gamma}$ is thermomajorized by any other state and thus any state can be converted to the Gibbs state. 
Conversely, the most excited state with $\v{p} = (0,\cdots,0,1)$ thermomajorizes all other states in $\Delta_{d}$ and any state can be converted from it. 

Theorem~\ref{thm:thermomajorisation} provides an efficient way of checking whether a transition from one state to another is feasible. 
However, it is also possible to characterize the entire set of states that can be attained from an initial state via thermal operations. 
We refer to such set with an initial state $\v{p}$ as the future thermal cone of $\v{p}$ and denote it with a symbol $\mathcal{T}_{+}(\v{p})$.
\begin{lem}[Lemma 12 of Ref.~\cite{Lostaglio2018elementarythermal}]
    \label{lem_extreme}
    Suppose $\v{p}\in\Delta_{d}$. 
    The future thermal cone $\mathcal{T}_{+}(\v{p})$ is a convex polytope with at most $d!$ extreme points each corresponding to permutations $\v \pi\in \S_d$.
    The extreme point $\v{p}^{\v{\pi}}$ corresponding to the permutation $\v{\pi}$ can be identified by first constructing its following thermomajorization curve $\mathcal{L}_{\v{p}^{\v{\pi}}}$:
    \begin{enumerate}
        \item The $x$-coordinates of the $\mathcal{L}_{\v{p}^{\v{\pi}}}$ elbow points are $x_k = \sum_{i=1}^k\gamma_{\pi^{-1}(i)}$ for $k \in\left\{1,\dots,d\right\}$.
        \item The $y$-coordinates of the $\mathcal{L}_{\v{p}^{\v{\pi}}}$ elbow points are $y_{k} = \mathcal{L}_{\v{p}}(x_k)$ for $k \in\left\{1,\dots,d\right\}$.
    \end{enumerate}
	Therefore, $\mathcal{L}_{\v{p}^{\v{\pi}}}$ is the piecewise-linear curve that connects $(0,0)$ and elbow points $\{(x_k,y_k)\}_{k=1}^{d}$, which all lie on the curve $\mathcal{L}_{\v{p}}$. 
\end{lem}
The thermomajorization relation between the initial state $\v{p}$ and extreme states $\v{p}^{\v{\pi}}$ are also known as tight-thermomajorization, since all elbow points of $\mathcal{L}_{\v{p}^{\v{\pi}}}$ are on the curbe $\mathcal{L}_{\v{p}}$.
Furthermore, by construction, each state $\v{p}^{\v{\pi}}$ has the $\beta$-ordering $\v{\pi}$.
Lemma~\ref{lem_extreme} enables us to characterize the future thermal cone of $\v p$ by constructing states $\v p^{\,\v \pi}$ for each $\v \pi \in \mathcal{S}_{d}$, and by taking their convex hull. Here, we provide the construction of $\T_{+}(\v p)$ as a simple corollary of Lemma~\ref{lem_extreme}:
\begin{thm}[Future thermal cone]
\label{thm_futurefinite}
The future thermal cone of a $d$-dimensional energy-incoherent state $\v p$ is given by
\begin{equation}
\T_{+}(\v p) = \operatorname{conv}[\{\v p^{\v \pi}, \v \pi\in\S_d\}] \, .
\end{equation}
\end{thm}
See Fig.~\ref{F:app-example} (b) for example future thermal cones of two three-level states. In this example, $\v{p}\succ_{\beta}\v{q}$ and $\T_{+}(\v{q})\subset\T_{+}(\v{p})$; these two relations are indeed equivalent in general. The equivalence directly follows from the transitivity of the preorder $\succ_{\beta}$, i.e. if $\v{p}\succ_{\beta}\v{q}$ and $\v{q}\succ_{\beta}\v{r}$, it is always true that $\v{p}\succ_{\beta}\v{r}$; hence, $\v{r}\in\T_{+}(\v{q})$ implies $\v{r}\in\T_{+}(\v{p})$.

It is worth mentioning that $\mathcal{T}^{\,\beta}_{+}(\v p)$ can also be constructed by identifying the entire set of extremal Gibbs-preserving stochastic matrices, as discussed in~\cite{Lostaglio2018elementarythermal,mazurek2018decomposability,mazurek2019thermal}. However, if we are only interested in the future thermal cones, following Lemma~\ref{lem_extreme} and Theorem~\ref{thm_futurefinite} is more efficient.

\section{Convexity of the subspace non-entanglable set $\mathbbm{NE}$}\label{App:convexity-nonentaglable}

\begin{figure*}
    \centering
    \includegraphics{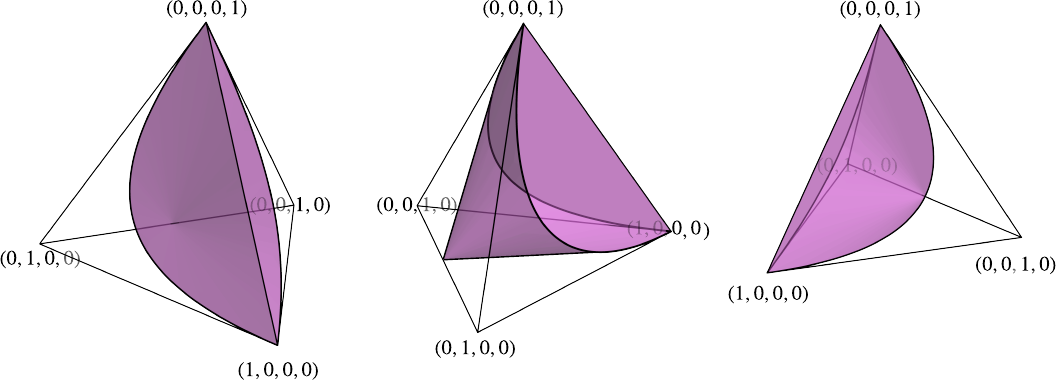}
    \caption{\textbf{Subspace non-entanglable set $\mathbbm{NE}$}. Geometry of the set of qubit-qubit states that cannot get entangled using energy-preserving unitaries. The extreme points of this set include the trivial ones, given by $(1,0,0,0), (0,0,0,1)$ and $(0,\nicefrac{1}{2},\nicefrac{1}{2},0)$, as well as a continuous family of boundary points specified by the solution to the quadratic equation $-(p_2-p_3)^2 +4p_1 (1-p_1-p_2-p_3) =0$ when solved for $p_3$.  }
    \label{Fig-non-entanglable-set}
\end{figure*}

For a bipartite two-qubit state, we now show that the subspace non-entanglable set $\mathbbm{NE}$ is convex. 
The set $\mathbbm{NE}$ is defined (in Definition~\ref{def:SNS}) as the set of states satisfying $f(\v{p})\geq0$, where $f$ is a function given in Eq.~\eqref{Eq:entanglement-measure}:
\begin{equation}\label{Eq:function-f}
 f(\v p)= 4p_1(1-p_1-p_2-p_3)-
 (p_2-p_3)^2.
\end{equation}
To show that the set of states defined by $f(\v p) \geq 0$ is convex, we need to demonstrate that the function $f$ is, counterintuitively, concave. Our task is to verify that for any two arbitrary points $P_1 = (x_1, y_1, z_1)$ and $P_2 = (x_2, y_2, z_2)$, the following condition holds:
\begin{equation}\label{Eq:f-convexity}
    f[\lambda(x_1,y_1,z_1) + (1-\lambda)(x_2,y_2,z_3)] \geq \lambda f(x_1,y_1,z_1) + (1-\lambda)f(x_2,y_2,z_2),
\end{equation}
for $\lambda \in [0,1]$. To make the calculation less cumbersome, we can prove convexity by considering two arbitrary boundary points. To do this, we first find the roots of Eq.~\eqref{Eq:function-f} with respect to $p_3$:
\begin{equation}\label{Eq:f-solutions}
     f(p_1,p_2,p_3)= -p^2_3+2p_3(p_2-2p_1)-[p_2(4p_1+p_2)-4p_1(1-p_1)]=0,
\end{equation}
whose solutions are simply given by
\begin{align}
    p^{(\pm)}_3 = -2p_1+ p_2\pm2\sqrt{p_1-2p_1 p_2}, \label{Eq:solution-a}
\end{align}
By focusing on the boundary conditions, we can simplify the analysis concerning convexity. Specifically, when considering two arbitrary boundary points, the right-hand side of Eq.~\eqref{Eq:f-convexity} is zero. This reduces the problem to verifying whether the left-hand side of this equation is less than or equal to zero. Given that there are two solutions for Eq.~\eqref{Eq:function-f}, we need to demonstrate that all possible combinations of boundary points $P_i = (x_1, y_1, p^{(\pm)}_3)$ for $i \{1,2\}$, the following inequalities are satisfied:
\begin{align}
    f[\lambda(x_1,y_1,p^{(\pm)}_3) + (1-\lambda)(x_2,y_2,p^{(\pm)}_3)] &\geq 0, \label{Eq:inequality-1}\\
    f[\lambda(x_1,y_1,p^{(\pm)}_3) + (1-\lambda)(x_2,y_2,p^{(\mp)}_3)] &\geq 0,\label{Eq:inequality-2}
\end{align}
Furthermore, we must impose the constrain $y_i \leq 1/2$ for $i\in\{1,2\}$ to ensure that the values within the square root in \eqref{Eq:solution-a} are real. By rearranging and rewriting the inequalities above, we observe that both inequalities in \eqref{Eq:inequality-1}, as well as inequalities \eqref{Eq:inequality-2}, are equivalent. Consequently, our analysis boils down to demonstrating the validity of the following two inequalities:
\begin{align}
    4\lambda(1-\lambda)[x_1(1-2y_2)+x_2(1-2y_1)\pm2\sqrt{(x_1-2x_1 y_1)(x_2-2x_2 y_2)}] &\geq 0.
\end{align}
Since the function $4\lambda(1-\lambda)$ is concave and positive on $\lambda\in[0,1]$ interval, we only need to prove that the second term is positive. By setting $a = \sqrt{x_1(1-2y_2)}$ and $b = \sqrt{x_2(1-2y_1}$ we see that the entire expression is reduced to
\begin{align}
    4\lambda(1-\lambda)(a\pm b)^2 &\geq 0,
\end{align}
which completes the proof of concavity of function $f$.

\section{Proof of Theorem~\ref{lem:criterion_TE}}\label{App:lemma_proof}
We first establish the existence of a subspace entanglable $\beta$-ordering $(2,1,3,4)$ state for all thermally entanglable states. 
Yet, the subspace entanglability of the extreme point $\v{p}^{\star}$ is not guaranteed from Lemma~\ref{Lem:extreme-beta-ordering} yet.
\begin{lem}[Critical $\beta$-ordering]\label{Lem:extreme-beta-ordering}
	For any state $ \v p \in \mathbbm{TE}$, there exists a state $ \v q $, such that $ \v q\in  \mathcal{T}^{\mathrm{ENT}}_{+}(\v p) $ and the $ \beta $-order $ \v{\pi}_{\v q} = (2,1,3,4) $.
\end{lem}
\begin{proof}
	Suppose that the Lemma is true for all $\v{p}\in\mathbbm{E}$.
	Then for any state $\v{r}\in\mathbbm{TE}$, there exists $\v{p}\in\mathbbm{E}$ and $\v{q}\in\mathcal{T}_{+}^\mathrm{ENT}(\v{p})\subset\mathcal{T}_{+}^\mathrm{ENT}(\v{r})$ with $\v{\pi}_{\v{q}} = (2,1,3,4)$.
	Hence, it is sufficient to prove the Lemma for $\v{p}\in\mathbbm{E}$.

	Without loss of generality, assume that $ p_{2}\geq p_{3}$.
	From the assumption $\v{p}\in\mathbbm{E}$, we have $f(\v p)<0$ for $f$ defined in Eq.~\eqref{Eq:entanglement-measure}. 
	Now, consider another state $ \v r $ such that $ \v r \in \mathcal{T}_{+}(\v p) $ and 
	\begin{align}\label{eq:rstate_2max}
		\v r = \left(p_{1}-\delta_{1},p_{2} +\sum_{i=1,3,4}\delta_{i},p_{3} - \delta_{3},p_{4} - \delta_{4}\right) ,
	\end{align} 
	for some $ \delta_{1},\delta_{3},\delta_{4} \geq 0$.
	We can observe that $ f(\v r) $ is decreasing monotonically as $ \delta_{1}, \delta_{3}, \delta_{4} $ increases. 
	If $ (\v{\pi}_{\v p})_{1} \neq 2 $, there always exists a state $\v{r}$ in the form of Eq.~\eqref{eq:rstate_2max} and $(\v{\pi}_{\v{r}})_{1} = 2$.
	From the monotonicity of $f(\v{r})$, this state $\v{r}\in\mathcal{T}^{\mathrm{ENT}}_{+}(\v p)$.
	We prove the rest of the Lemma by showing there exists $\v{q}\in\mathcal{T}^{\mathrm{ENT}}_{+}(\v r)$ and use $\mathcal{T}^{\mathrm{ENT}}_{+}(\v r)\subset\mathcal{T}^{\mathrm{ENT}}_{+}(\v p)$.
	Hence, we assume that $(\v{\pi}_{\v{p}})_{1} = 2$.

	Now, suppose that the state $ \v p $ has $ p_{4} \geq p_{1}e^{-2\beta E} $. 
	Then a state $ \v r $ such that $ r_{2} =p_{2} $, $ r_{3} = p_{3} $, and 
	\begin{align}
		r_{1} = \left(1-e^{-2\beta E}\right)p_{1} + p_{4},\quad r_{4} = e^{-2\beta E}p_{1},
	\end{align}
	is in $ \mathcal{T}_{+}(\v p) $ and satisfies $  r_{1}e^{-2\beta E}\geq r_{4} $.
	From direct calculation,
	\begin{align}
		f(\v r) = f(\v p) - 4\left(1-e^{-2\beta E}\right)p_{1}\left(p_{4} - e^{-2\beta E}p_{1}\right)\leq f(\v p),
	\end{align}
	i.e. $ \v r \in \mathcal{T}^{\mathrm{ENT}}_+(\v p) $ whenever $\v{p}\in\mathbbm{E}$.
	
	Therefore, one can always find $ \v q\in  \mathcal{T}^{\mathrm{ENT}}_{+}(\v p)$ whose $ \beta $-ordering is either \emph{(i)} $ \pi_{\v q} = (2,3,1,4) $, \emph{(ii)} $ \pi_{\v q} = (2,1,4,3) $, or \emph{(iii)} $ \pi_{\v q} = (2,1,3,4) $.
	We will show that cases \emph{(i)} and \emph{(ii)} again imply the existence of $ \v q\in  \mathcal{T}^{\mathrm{ENT}}_{+}(\v p)$ corresponding to the case \emph{(iii)}, which proves the Lemma.
	
	Case \emph{(i)}: $ \pi_{\v p} = (2,3,1,4) $ with a strict ordering $p_{3}>p_{1}e^{-\beta E}$. 
	It is always possible to find $ \v p' \in \mathcal{T}_{+}(\v p) $, such that $ \pi_{\v p'} = (2,3,1,4) $ and $ f(\v p') = 0 $ by a convex combination $\v p' = (1-\lambda)\v p + \lambda \v \gamma$, where $\lambda\in[0,1]$ and $\gamma$ is the Gibbs state.
	Now assume that $ \v q \in \mathcal{T}_{+}(\v p') $ is a state with $ \pi_{\v q} = (2,1,3,4) $ and
	\begin{align}
		\v q = \left(p'_{1} + \delta ,p'_{2},p'_{3} - \delta,p'_{4}\right),
	\end{align}
	with some $ \delta> 0 $.
	Such a state always exists since $p'_{3} > p'_{1}e^{-\beta E}$.
	The entanglement witness gives
	\begin{align}
		f(\v q) &=  4\left(p'_{1}  +\delta\right)\left(p'_{4}\right)-\left(p'_{2}-p'_{3} + \delta\right)^{2}\nonumber\\
		&= -\delta^{2} + 2\left(p'_{3}+2p'_{4}-p'_{2}\right)\delta,\label{eq:Nq_1st_exp}
	\end{align}
	where we used $ f(\v p') = 0 $ for the second equality.
	If $p'_{1} = 0$, we have $p'_{4}=0$ from the $\beta$-order of $\v p'$ and $p'_{2} = p'_{3}$ from $f(\v p') = 0$; hence, $f(\v q) < 0$.
	When $p'_{1}\neq 0$, we use $ p'_{4} = (p'_{2} - p'_{3})^{2}/(4p'_{1}) $ to get
	\begin{align}
		f(\v q) = -\delta^{2} - 2\left(p'_{2}-p'_{3}\right)\left(1-\frac{p'_{2}-p'_{3}}{2p'_{1}}\right)\delta.
	\end{align}
	Recall $ 4(p'_{1})^{2} \geq 4p'_{1}p'_{4} = (p'_{2}-p'_{3})^{2} $ implying $ 1-\frac{p'_{2}-p'_{3}}{2p'_{1}} \geq 0 $ and thus $ f(\v q)< 0 $. 
	
	Case \emph{(ii)}: $ \pi_{\v p} = (2,1,4,3) $. There exists $ \v p' \in \mathcal{T}_{+}(\v p) $, such that $ \pi_{\v p'} = (2,1,4,3) $ and $ f(\v p') = 0 $ and also $ \v q \in \mathcal{T}_{+}(\v p') $, such that 
	\begin{align}
		\v q = \left(p'_{1} ,p'_{2},p'_{3} + \delta,p'_{4} - \delta\right),
	\end{align}
	with $ \pi_{\v q} = (2,1,3,4) $.
	Observe that
	\begin{align}
		f(\v q) &=  4\left(p'_{1} \right)\left(p'_{4} - \delta\right)-\left(p'_{2}-p'_{3} - \delta\right)^{2}\nonumber\\
		&= -\delta^{2} - 2\left(2p'_{1} - p'_{2} + p'_{3}\right)\delta < 0
	\end{align}
	for any $ \delta>0 $.
\end{proof}

Now we use geometric arguments to prove the remaining part of Theorem~\ref{lem:criterion_TE}.
To do this, we define generalized boundary sets. 
\begin{defn}[$d$-dimensional boundaries]
	Let $X$ be a convex set and $C_{d}(X)$ be a set of all $(d+1)$-dimensional convex subsets $Y\subset X$.
	$d$-dimensional boundary of $X$ is defined as
 \begin{align}\label{eq:generalized_Bd_def}
		\Bd_{d}(X) \coloneqq \left\{x\, \big| \, \not\exists Y\in C_{d}(X),\ \text{such that}\ x\in \mathrm{int}(Y) \right\}.
	\end{align}
	Intuitively, $\Bd_{d}(X)$ cannot be written as a convex combination of $d+2$ extreme points of $X$ with nonzero weights; hence,	$\Bd_{0}(X)$ is a set of extreme points of $X$, while $\Bd_{1}(X)$ is a set of points of the edges of $X$.
	Furthermore, if $X$ is $D$-dimensional, $\Bd_{D-1}(X)$ reduces to the conventional boundary, $ \partial X$ 
\end{defn}
Note that $\Bd_{d}(X)\subset\Bd_{d+1}(X)$ for any $d\geq 0$: if it is impossible write a state $x$ as a convex combination of any $d+2$ extreme points, it is also impossible to write it as a convex combination of any $d+3$ extreme points. 

Using this definition, we establish a helpful proposition.
\begin{prop}\label{Lem:extreme_conv_points}
	Suppose that a convex set $A$ is a subset of a convex set $B$ and their boundaries interesect, i.e. $A\subset B$ and $\partial A \cap \partial B \neq \emptyset$. 
	Then, for all $d$-dimensional boundaries,
	\begin{align}\label{eq:lem_exconvpoints_eq1}
		A \cap \Bd_{d}(B) \subset \Bd_{d}(A).
	\end{align}
	If we have additional condition $[\Bd_{d}(A)\setminus\Bd_{d-1}(A)] \cap \Bd_{d}(B) = \emptyset$, for some $d$, we find
	\begin{align}\label{eq:lem_exconvpoints_eq2}
		A \cap \Bd_d(B) \subset \Bd_{d-1}(A).
	\end{align}
\end{prop}
\begin{proof}
	For the first part of the proposition, suppose that there exists a point $x\notin\Bd_{d}(A)$ such that $x\in A \cap \Bd_{d}(B)$.
	Then there exists a convex subset $Y$ of $A$ having $d+2$ extreme points and containing $x$ in $\mathrm{int}(Y)$.
	Yet, this set $Y$ is also a convex subset of $B$ having the same property, which contradicts $x\in A \cap \Bd_{d}(B)$.
	
	The second part of the proposition follows directly from combining Eq.~\eqref{eq:lem_exconvpoints_eq1} and the additional assumption $[\Bd_{d}(A)\setminus\Bd_{d-1}(A)] \cap \Bd_{d}(B) = \emptyset$.
\end{proof}

Now we proceed to prove Theorem~\ref{lem:criterion_TE}, which states that $\v{p}^{\star}\in\mathbbm{E}$ whenever $\v{p}\in\mathbbm{TE}$.
First, consider mixtures of the initial state $\v{p}$ and the Gibbs state $\v{\gamma}$ that can be written as $\v{p}^{(x)} = (1-x)\v{p} + x\v{\gamma}$ for some $x\in[0,1]$. 
Future thermal cones between these states have a hierarchy $\mathcal{T}_{+}(\v{p}^{(y)})\subset\mathcal{T}_{+}(\v{p}^{(x)})$ whenever $y\geq x$.
Then there exists a coefficient $a\in(0,1)$, such that 
\begin{align}
	\mathcal{T}_{+}\left(\v{p}^{(a)}\right) \subset \mathbbm{NE},\quad \partial\left(\mathcal{T}_{+}\left(\v{p}^{(a)}\right)\right) \cap \partial\left(\mathbbm{NE}\right) \neq \emptyset,
\end{align}
since $	\mathcal{T}_{+}(\v{p}) \not\subset \mathbbm{NE}$, while $	\mathcal{T}_{+}(\v{\gamma}) \subset \mathrm{int}(\mathbbm{NE})$.
The Gibbs state $\v{\gamma}$ is always full-rank and so is $\v{p}^{(a)}$, i.e. $\v{p}^{(a)}\in\mathrm{int}(\Delta_{4})$ and $\mathcal{T}_{+}(\v{p})\subset\mathrm{int}(\Delta_{4})$.

Now we show that $[\Bd_{1}(\mathcal{T}_{+}(\v{p}^{(a)}))\setminus\Bd_{0}(\mathcal{T}_{+}(\v{p}^{(a)}))]\cap \partial (\mathbbm{NE}) = \emptyset$.
Let $\v{q}$ be a state on the edge of $\mathcal{T}_{+}(\v{p}^{(a)})$ but not on the vertex of it, i.e. $\v{q}\in[\Bd_{1}(\mathcal{T}_{+}(\v{p}^{(a)}))\setminus\Bd_{0}(\mathcal{T}_{+}(\v{p}^{(a)}))]$.
Since $\mathbbm{NE}$ is convex, if $\v{q}\in\partial(\mathbbm{NE})$, the entire edge containing $\v{q}$ should be included within $\partial(\mathbbm{NE})$.
Recall that $\mathcal{T}_{+}(\v{p}^{(a)})$ is a polyhedron whose edges are defined as segments of lines on which two of the populations are fixed, e.g. 
$\v{q} = (p^{(a)}_{1},p^{(a)}_{2},t,1-p^{(a)}_{1}-p^{(a)}_{2}-t)$ and for some $t\in(t_{1},t_{2})$.
Then the entire edge can be written as a set of states 
\begin{align}\label{eq:edge_in_NE_boudnary}
	\left\{\left(p^{(a)}_{1},p^{(a)}_{2},t,1-p^{(a)}_{1}-p^{(a)}_{2}-t\right)\, \big\vert\, t \in [t_{1},t_{2}]\right\} \subset \partial(\mathbbm{NE}).
\end{align}
Yet, it is impossible to satisfy $f(p^{(a)}_{1},p^{(a)}_{2},t,1-p^{(a)}_{1}-p^{(a)}_{2}-t) = 0$ for a finite range $ t \in [t_{1},t_{2}]$, which contradicts Eq.~\eqref{eq:edge_in_NE_boudnary}.
The same can be proven for any two fixed populations. 
Likewise, $[\Bd_{2}(\mathcal{T}_{+}(\v{p}^{(a)}))\setminus\Bd_{1}(\mathcal{T}_{+}(\v{p}^{(a)}))]\cap \partial (\mathbbm{NE}) = \emptyset$: if any point on $[\Bd_{2}(\mathcal{T}_{+}(\v{p}^{(a)}))\setminus\Bd_{1}(\mathcal{T}_{+}(\v{p}^{(a)}))]$ is included in $\partial (\mathbbm{NE})$, at least one entire edge of $\mathcal{T}_{+}(\v{p}^{(a)})$ should also be included, which is not possible. 

Using Proposition~\ref{Lem:extreme_conv_points} twice with $A = \mathcal{T}_{+}(\v{p}^{(a)})$ and $B = \mathbbm{NE}$, we have 
\begin{align}
	\mathcal{T}_{+}(\v{p}^{(a)})\cap\partial(\mathbbm{NE}) = \mathcal{T}_{+}(\v{p}^{(a)})\cap\partial(\mathbbm{NE}) \subset \Bd_{0}(\mathcal{T}_{+}(\v{p}^{(a)})),
\end{align}
i.e. only extreme points of $\mathcal{T}_{+}(\v{p}^{(a)})$ can be included in the set $\partial(\mathbbm{NE})$. 

We finally use Lemma~\ref{Lem:extreme-beta-ordering}.
Since $\v{p}^{(a)}\in\partial(\mathbbm{TE})$, there must exist $\v{q}\in\mathcal{T}_{+}(\v{p}^{(a)})$ with $\beta$-ordering $(2,1,3,4)$ that is in $\partial(\mathbbm{E}) = \partial(\mathbbm{NE})$.
However, only extreme points of $\mathcal{T}_{+}(\v{p}^{(a)})$ can be in $\partial(\mathbbm{NE})$, implying that the $(2,1,3,4)$ extreme point 
\begin{align}
	(1-a)\v{p}^{\star} + a\v{\gamma} \in \partial(\mathbbm{NE}).
\end{align}
From $\v{\gamma}\in\mathrm{int}(\mathbbm{NE})$, we obtain $\v{p}^{\star}\in\mathbbm{E}$, which concludes the proof.

\section{Thermal initial states}\label{A:thermal-state}

In this section, we derive the critical temperatures above or below which the initially thermal state become entanglable. 
We denote the inverse temperature of the system as $\beta_{\ms{S}}$ and that of the environment as $\beta$.
To simplify the notation, we use Boltzmann weights $\Delta_{\ms{S}} \coloneqq e^{-\beta_{\ms{S}}E}$, $\Delta \coloneqq e^{-\beta E}$, , and partition functions $Z_{\ms{S}} \coloneqq (1 + \Delta_{\ms{S}})^{2}$, $Z \coloneqq (1+ \Delta)^{2}$.
Then the initial state can be written as
\begin{align}
	\v{p} = \frac{1}{Z_{\ms{S}}}\left(1,\Delta_{\ms{S}}, \Delta_{\ms{S}}, \Delta_{\ms{S}}^{2}\right).
\end{align}

When $\beta_{S}>\beta$, i.e. when the system is cooler than the environment, the $\beta$-ordering of $\v{p}$ is $(1,2,3,4)$.
The extreme point $\v{p}^{\star}$ of $\mathcal{T}_{+}(\v{p})$, corresponding to the ordering $\v{\pi}^{\star} = (2,1,3,4)$ can be easily obtained by maximally exchanging populations (referred to as the $\beta$-swapping in Refs.~\cite{Lostaglio2018elementarythermal, PRXQuantum.4.040304, Son2024_CETO}) between first two levels $\ket{00}$ and $\ket{01}$,
\begin{align}\label{eq:pstar_thermal_cooler}
	\v{p}^{\star} = \frac{1}{Z_{\ms{S}}}\left(1-\Delta+\Delta_{\ms{S}}, \Delta, \Delta_{\ms{S}}, \Delta_{\ms{S}}^{2} \right).
\end{align}
We want the entanglement measure to be negative, i.e. 
\begin{align}
	f(\v{p}^{\star}) = \frac{1}{Z_{\ms{S}}^{2}}\left(4\Delta_{\ms{S}}^{2}\left[1-(\Delta-\Delta_{\ms{S}})\right] - (\Delta-\Delta_{\ms{S}})^{2}\right)<0.
\end{align}
Solving this inequality, we obtain 
\begin{align}\label{eq:Delta_C1}
	\Delta > \Delta_{C_{1}}\coloneqq \Delta_{\ms{S}}\qty(1+2\sqrt{1+\Delta_{\ms{S}}^{2}} - 2\Delta_{\ms{S}}),  
\end{align}
equivalent to Eq.~\eqref{eq:betaC1_explicit} in Section~\hyperref[Sec:future-thermal-entanglement]{V-B}.

When the system temperature is low, that is, when $\Delta_{\ms S} \ll 1$, Eq.~\eqref{eq:Delta_C1} can be written with the Taylor expansion
\begin{align}
	\Delta_{C_{1}} = 3\Delta_{\ms{S}} + O(\Delta_{\ms{S}}^{2}),
\end{align}
yielding the approximation $\beta_{C_{1}} E \simeq \beta_{\ms S} E-\log3$, up to the leading order.

Now, consider the other case where $\beta_{\ms{S}} < \beta$; the system is now hotter than the surrounding environment, and the $\beta$-order of the initial state $\mathbf{p}$ is now $(4, 3, 2, 1)$. 
Thermomajorization relations reveal that the extreme point $\v{p}^{\star}$ now has two distinct forms depending on the environment temperature $\beta$. 
If $1 > \Delta + \Delta^{2}$, that is, $\beta E \geq -\log [\frac{\sqrt{5}-1}{2}]$, the extreme state is
\begin{align}\label{eq:pstar_thermal_hotter_1}
	\v{p}^{\star} = \frac{1}{Z_{\ms{S}}}\left(1+(\Delta_{\ms{S}} - \Delta)(1+\Delta), \Delta_{\ms{S}}(1+\Delta_{\ms{S}} - \Delta), \Delta, \Delta^{2}\right).
\end{align}
On the other hand, when $\beta E < -\log [\frac{\sqrt{5}-1}{2}]$,
\begin{align}\label{eq:pstar_thermal_hotter_2}
	\v{p}^{\star} = \frac{1}{Z_{\ms{S}}}\left(\frac{\Delta_{\ms{S}}}{\Delta}, \Delta_{\ms{S}}(1+\Delta_{\ms{S}} - \Delta), 1+\Delta_{\ms{S}}+\Delta_{\ms{S}}\Delta-\Delta-\frac{\Delta_{\ms{S}}}{\Delta}, \Delta^{2}\right).
\end{align}
For the latter case Eq.~\eqref{eq:pstar_thermal_hotter_2}, the entanglement measure is evaluated as
\begin{align}
	f(\v{p}^{\star}) = \frac{1}{Z_{\ms{S}}^{2}\Delta^{2}}\left(4\Delta_{\ms{S}}\Delta^{3} - (\Delta_{\ms{S}}-\Delta)^{2}(1+\Delta(\Delta_{\ms{S}}-\Delta))^{2}\right).
\end{align}
However, in the range where Eq.~\eqref{eq:pstar_thermal_hotter_2} is valid, $1\geq\Delta_{\ms{S}}>\Delta>\frac{\sqrt{5}-1}{2}$, entanglement measure is always positive and there is no critical inverse temperature $\beta_{C_{2}}$.
For the former case Eq.~\eqref{eq:pstar_thermal_hotter_1} with $1\geq\Delta_{\ms{S}}>\Delta$ and $\frac{\sqrt{5}-1}{2}\geq\Delta>0$,
\begin{align}\label{eq:fpstar_thermal_hotter_1}
	f(\v{p}^{\star}) = \frac{1}{Z_{\ms{S}}^{2}}\left[-(\Delta_{\ms{S}}-\Delta)^{2}(1+\Delta_{\ms{S}})^{2} + 4\Delta^{2}\left(1+\Delta_{\ms{S}} - (1-\Delta_{\ms{S}})\Delta + \Delta^{2}\right)\right]
\end{align}
and the solution for $f(\v{p}^{\star}) = 0$ can be obtained by solving this quartic equation.
Unfortunately, there is no simple explicit solution for this equation.
To approximate the critical value, let us assume that $\Delta_{\ms{S}}\ll 1$, i.e. $\beta_{\ms{S}} E \gg1$.
Ignoring the $Z_{\ms{S}}$ factor in Eq.~\eqref{eq:fpstar_thermal_hotter_1} and taking only the leading order terms of $\Delta$ and $\Delta_{\ms{S}}$, we get the approximation
\begin{align}
	Z_{\ms{S}}^{2}f(\v{p}^{\star}) = 3\Delta^{2} + 2\Delta\Delta_{\ms{S}} - \Delta_{\ms{S}}^{2} + O(\Delta_{\ms{S}}^{3}),
\end{align}
and the solution of $f(\v{p}^{\star}) = 0$ is approximately $\Delta\simeq \frac{\Delta_{\ms{S}}}{3}$.
In other words, $\beta_{C_{2}} E \simeq \beta_{\ms{S}} E+ \log3$.

\section{Catalytic advantages on entanglement generation - a toy example}\label{App:catalytic-advantages}

We now construct a simple toy example to demonstrate that if a catalyst is allowed, entanglement can be generated from states that cannot achieve this under thermal operations without a catalyst. For simplicity, we assume $\beta = 0$ and consider a bipartite state
\begin{equation}
\rho_{\AB} = \frac{2}{5}\ketbra{00}{00}+ \frac{1}{4}\ketbra{01}{01}+ \frac{33}{100}\ketbra{10}{10}+ \frac{1}{50}\ketbra{11}{11},
\end{equation}
with local Hamiltonians $H_{\ms X} = \ketbra{1}{1}$ for $\ms X \in \{\ms A, \ms B\}$. The initial bipartite state $\rho_{\AB}$ belongs to the thermally non-entanglable set, i.e., $\v p = \mathrm{eig}(\rho) \in \mathbbm{TNE}$. Hence, at $\beta = 0$, there is no thermal operation capable of entangling $\rho_{\AB}$.

Let us now consider a two-dimensional catalyst, prepared in a state
\begin{equation}
     \omega_{\ms{C}}  = \frac{73}{100}\ketbra{0}{0}+\frac{27}{100}\ketbra{1}{1},
\end{equation}
and described described by a Hamiltonian $H_{\ms{C}} = \ketbra{1}{1}$, such that the joint Hamiltonian of the composite system $\ms{ABC}$ is given by
\begin{equation}
    H_{\ms{ABC}} =
\begin{pNiceMatrix}[margin]
\Block[fill=boxcolor!,rounded-corners]{1-1}{} 0 & 0 & 0 & 0 & 0 & 0 & 0 & 0 \\
 0 & \Block[fill=block2c!,rounded-corners]{3-3}{} 1 & 0 & 0 & 0 & 0 & 0 & 0 \\
 0 & 0 & 1 & 0 & 0 & 0 & 0 & 0 \\
 0 & 0 & 0 & 1 & 0 & 0 & 0 & 0 \\
 0 & 0 & 0 & 0 & \Block[fill=block3c!,rounded-corners]{3-3}{} 2 & 0 & 0 & 0 \\
 0 & 0 & 0 & 0 & 0 & 2 & 0 & 0 \\
 0 & 0 & 0 & 0 & 0 & 0 & 2 & 0 \\
 0 & 0 & 0 & 0 & 0 & 0 & 0 & \Block[fill=block4c!,rounded-corners]{1-1}{} 3 \\
\end{pNiceMatrix}.
\end{equation}
where the colors in each block highlight the degenerate energy subspaces. We assume that the total system evolves according to the following energy-preserving unitary transformation
\begin{equation}
U = \begin{pNiceMatrix}[margin]
\Block[fill=boxcolor!,rounded-corners]{1-1}{} 1 & 0 & 0 & 0 & 0 & 0 & 0 & 0 \\
 0 & \Block[fill=block2c!,rounded-corners]{3-3}{} 1 & 0 & 0 & 0 & 0 & 0 & 0 \\
 0 & 0 & 0 & 1 & 0 & 0 & 0 & 0 \\
 0 & 0 & 1 & 0 & 0 & 0 & 0 & 0 \\
 0 & 0 & 0 & 0 & \Block[fill=block3c!,rounded-corners]{3-3}{} 0 & 1 & 0 & 0 \\
 0 & 0 & 0 & 0 & 0 & 0 & 1 & 0 \\
 0 & 0 & 0 & 0 & 1 & 0 & 0 & 0 \\
 0 & 0 & 0 & 0 & 0 & 0 & 0 & \Block[fill=block4c!,rounded-corners]{1-1}{} 1 \\
\end{pNiceMatrix},
\end{equation} 
which maps the composite system to $\sigma_{\ms{ABC}} := U(\rho_{\ms{AB}} \otimes \omega_{\ms{C}})U^{\dagger}$, where the reduced states of $\AB$ become
\begin{align}\label{Eq:finalstate-after-catalyst}
    \sigma_{\AB} = \frac{949}{2000}\ketbra{00}{00}+ \frac{613}{5000}\ketbra{01}{01}+ \frac{771}{2500}\ketbra{10}{10}+ \frac{189}{2000}\ketbra{11}{11}.
\end{align}
while the catalyst returns exactly to the same state as it started $\sigma_{\ms{C}}= \omega_{\ms{C}}$. 

Therefore, one can easily verify that the final state no longer belongs to $\mathbbm{TNE}$, i.e., $\v q := \operatorname{eig}(\sigma) \notin \mathbbm{TNE}$. This implies that a thermal operation capable of mapping $\sigma_{\AB}$ into an entangled state can be constructed.

\end{document}